\documentclass[sigconf]{acmart}

\usepackage{balance}
\usepackage{hyperref}
\usepackage{amsmath,amssymb,amsfonts}
\usepackage{algorithmic}
\usepackage{graphicx}
\usepackage{textcomp}
\usepackage{xcolor}
\usepackage{pifont}
\usepackage{amsthm} 
\usepackage{mathrsfs}
\usepackage{booktabs} 
\usepackage{url}  
\usepackage{graphicx}  
\usepackage[ruled,vlined,linesnumbered]{algorithm2e}
\usepackage{setspace}
\usepackage{bm}
\usepackage{floatpag} 
\usepackage{enumitem}
\usepackage{color}
\usepackage{subfigure}
\usepackage{enumitem}
\setlist{leftmargin=13pt}
\newtheorem{defn}{Definition}

\newtheorem{thm}{Theorem}

\newcommand{\bp}[1]{{\mathbb{P}}\left[{#1}\right]}

\newcommand{\argmaxa}{\operatornamewithlimits{argmax}}

\usepackage{fancyhdr}
\def\BibTeX{{\rm B\kern-.05em{\sc i\kern-.025em b}\kern-.08emT\kern-.1667em\lower.7ex\hbox{E}\kern-.125emX}}
    
\begin{document}


\title{Privacy-preserving Crowd-guided AI Decision-making\\ in Ethical Dilemmas}

\author{Teng Wang}
\affiliation{%
\institution{Xi'an Jiaotong University}
  \city{Xi'an}
  \state{Shaanxi}
  \country{China}
}
\email{wangteng0610@stu.xjtu.edu.cn}

\author{Jun Zhao}
\affiliation{%
  \institution{Nanyang Technological University}
  \country{Singapore}}
\email{junzhao@ntu.edu.sg}
  
\author{Han Yu}
\affiliation{%
  \institution{Nanyang Technological University}
  \country{Singapore}} 
\email{han.yu@ntu.edu.sg}
  
\author{Jinyan Liu}
\affiliation{%
  \institution{The University of Hong Kong}
  \country{Hong Kong}}
\email{jyliu@cs.hku.hk}

\author{Xinyu Yang}
\affiliation{%
  \institution{Xi'an Jiaotong University}
  \city{Xi'an}
  \state{Shaanxi}
  \country{China}}
\email{yxyphd@mail.xjtu.edu.cn}

\author{Xuebin Ren}
\affiliation{%
  \institution{Xi'an Jiaotong University}
  \city{Xi'an}
  \state{Shaanxi}
  \country{China}}
\email{xuebinren@mail.xjtu.edu.cn}

\author{Shuyu Shi}
\affiliation{%
 \institution{Nanjing University}
 \city{Nanjing}
 \state{Jiangsu}
 \country{China}}
\email{ssy@nju.edu.cn}

%
\renewcommand{\shortauthors}{~}

%
\begin{abstract}
With the rapid development of artificial intelligence (AI), ethical issues surrounding AI have attracted increasing attention. In particular, autonomous vehicles may face moral dilemmas in accident scenarios, such as staying the course resulting in hurting pedestrians or swerving leading to hurting passengers. To investigate such ethical dilemmas, recent studies have adopted preference aggregation, in which each voter expresses her/his preferences over decisions for the possible ethical dilemma scenarios, and a centralized system aggregates these preferences to obtain the winning decision. Although a useful methodology for building ethical AI systems, such an approach can potentially violate the privacy of voters since moral preferences are sensitive information and their disclosure can be exploited by malicious parties resulting in negative consequences. In this paper, we report a first-of-its-kind privacy-preserving crowd-guided AI decision-making approach in ethical dilemmas. We adopt the formal and popular notion of differential privacy to quantify privacy, and consider four granularities of privacy protection by taking voter-/record-level privacy protection and centralized/distributed perturbation into account, resulting in four approaches VLCP, RLCP, VLDP, and RLDP, respectively. Moreover, we propose different algorithms to achieve these privacy protection granularities, while retaining the accuracy of the learned moral preference model. Specifically, VLCP and RLCP are implemented with the data aggregator setting a universal privacy parameter and perturbing the averaged moral preference to protect the privacy of voters' data. VLDP and RLDP are implemented in such a way that each voter perturbs her/his local moral preference with a personalized privacy parameter. Extensive experiments based on both synthetic data and real-world data of voters' moral decisions demonstrate that the proposed approaches achieve high accuracy of preference aggregation while protecting individual voter's privacy.
\end{abstract}

%
%


%

%
\maketitle

\thispagestyle{fancy} \pagestyle{fancy}  

\section{Introduction}

Artificial intelligence (AI) is becoming an integral part of our daily lives and critical infrastructures. With the widespread applications of AI, ethical issues surrounding AI have become an important socio-technical challenge~\cite{wallach2008moral,awad2018moral}. One of the fundamental questions in AI ethics is how to allow humans to guide AI to make moral decisions when faced with ethical dilemmas.
An ethical dilemma is a situation in which any decision violates certain aspects of ethics~\cite{yu2018building}.
The following is a concrete example of an ethical dilemma that autonomous vehicles encounter \cite{greene2016our}. A self-driving car has a sudden mechanical failure and cannot brake in time. If the car were to continue its current trajectory, it would kill pedestrians but the passenger would be safe. Alternatively, if it were to swerve into a wall, the pedestrians would be safe but the passenger would be killed. In such a situation, an ethical dilemma arises when the AI must choose between the two alternatives. This is just one example of many ethical dilemmas AI technologies face \cite{shariff2017psychological}. To enable AI to deal with such situations, it is useful to aggregate opinions from the human society.

To explore moral dilemmas faced by autonomous vehicles, Bonnefon~\emph{et~al.}~\shortcite{bonnefon2016social} surveyed via Amazon Mechanical Turk whether people prefer to save many versus few lives, or the vehicle's passengers versus pedestrians. Their findings suggest that participants prefer others to buy autonomous vehicles which sacrifice their passengers for the greater good, but prefer for themselves to ride in autonomous vehicles which protect passengers at all costs. 

In a study similar to~\cite{bonnefon2016social}, Awad~\emph{et~al.}~\shortcite{awad2018moral} also gathered data about voters' decisions in scenarios where autonomous vehicles face ethical dilemmas, but at a much larger scale and with each voter's data containing more dimensions. Specifically, 
they built an online experimental platform named the Moral Machine~\cite{MoralMachine},
which collected 39.61 million decisions in ten languages from 4 million people in 233 countries and territories. The data collected from each voter are high-dimensional, which include preferences of saving many versus few lives, passengers versus pedestrians, the young versus the elderly, humans versus pets, and pedestrians who cross legally versus pedestrians who jaywalk, etc. With such a high volume of crowdsourced data, Awad~\emph{et~al.}~\shortcite{awad2018moral} summarized global and regional moral preferences as well as cross-cultural ethical variations.



Using the data collected from the Moral Machine,  
Noothigattu~\emph{et~al.} \shortcite{noothigattu2018voting} built a global moral preference model so as to guide automated ethical decision-making. Specifically, each voter's data are analyzed to infer her/his parameter of moral preference, and the moral preference of the society is obtained by averaging voters' parameters. Each voter's data consist of a number of records, with each record being the voter's preference in a given ethical dilemma scenario. In each scenario, no matter how the autonomous vehicle decides to act, some undesirable outcome will happen. A voter's preference in each scenario means that the voter prefers a decision made by the autonomous vehicle (e.g., staying the course and killing the crossing pedestrians) over the other alternative (e.g., swerving and killing the passengers).

As discussed above, existing studies on AI ethics (e.g.,~\cite{bonnefon2016social,awad2018moral,noothigattu2018voting}) directly analyze voters' data of moral decisions. However, such approaches may violate the privacy of voters. Since moral preferences are sensitive information and their disclosure may be exploited by malicious parties to produce negative consequences. 
One may wonder how the learned and aggregated moral preference model for the society can leak a voter's sensitive choices about the moral dilemmas. A recent work~\cite{fredrikson2015model} has shown that an adversary can successfully use a learned deep neural network model to infer users' sensitive facial information in the training dataset. Specifically, for an attacker which obtains the learned face recognition model, although it is difficult to infer all users' faces, the attacker may recover an image which is close to one user's face, resulting in a privacy breach of the particular user in the training dataset~\cite{fredrikson2015model}.

In this paper, to prevent the learned preference model of the society from leaking individual voter's sensitive information, we propose a differential privacy (DP)-based secure preference aggregation model to enable crowdsourced opinions to guide AI decision-making in ethical dilemmas without exposing sensitive privacy information. We adopt the formal notion of DP~\cite{dwork2014algorithmic,dwork2006our} to quantify privacy. Intuitively, by incorporating some noise, the output of an algorithm under DP will not change significantly due to the presence or absence of one voter's information in the dataset. 
\textbf{Contributions.} With the proposed approach, we make the following contributions in this paper.
\begin{itemize} 
    \item We quantify four granularities of privacy protection by combing voter-/record-level privacy protection and centralized/distributed perturbation, which are denoted by VLCP, RLCP, VLDP, and RLDP. We further propose different algorithms to achieve them. Specifically, to achieve VLCP and RLCP, the aggregator adds Laplace noise to the average preference parameter of all voters. VLDP and RLDP are achieved by having each voter adding Laplace noise to her/his local moral preference parameter with a personalized privacy parameter. Moreover, we also propose to achieve RLDP by perturbing the objective function when learning each voter's preference parameters, which achieves higher accuracy than the addition of Laplace noise.
    \item We conduct extensive experiments on both synthetic datasets and a real-world dataset extracted from the Moral Machine. The results demonstrate that our algorithms can achieve high accuracy while ensuring strong privacy protection.
\end{itemize}
To the best of our knowledge, this is the first research on privacy issues in studies of human-guided ethical AI decision-making. More specifically, we proposed privacy-preserving mechanisms to address voters' privacy in computing a society's moral preference.

\textbf{Organization.} The remainder of the paper is organized as follows. Section~\ref{sec:related} reviews the related studies of AI ethics and privacy protection. In Section~\ref{sec:preliminaries}, we discuss the preliminaries of differential privacy and formalize the research problem. Section~\ref{sec:mechanisms} presents the proposed privacy-preserving crowd-guided ethical AI decision-making algorithms. In Section~\ref{sec:experiments}, we conduct extensive experiments to evaluate the effectiveness of our algorithms. Section~\ref{sec:discussion} provides discussions and future directions. Section~\ref{sec:conclusion} concludes the paper.

\section{Related Work}\label{sec:related}

The widespread adoption of AI has made it pertinent to address the ethical issues surrounding this technology.


Greene~\textit{et~al.}~\shortcite{greene2016embedding} advocated solving AI ethical issues via preference aggregation, in which each voter expresses her/his preferences over the possible decisions, and a centralized system aggregates these preferences to obtain the winning decision. Conitzer~\textit{et~al.}~\shortcite{conitzer2017moral} discussed the idea of collecting a labeled dataset of moral dilemmas represented as lists of feature values, and then leveraging machine learning techniques to learn to classify actions as morally right or wrong. However, no actual data collection or analysis was presented. Using a large-scale dataset of voters' ethical decisions collected by the Moral Machine~\shortcite{awad2018moral}, Noothigattu~\textit{et~al.}~\shortcite{noothigattu2018voting} built a moral preference model for each voter and averaged these models to obtain the moral preference of the whole society. Hammond and Belle~\shortcite{hammond2018deep} utilized tractable probabilistic learning to induce models of moral scenarios and blameworthiness automatically from datasets of human decision-making, and computed judgments tractably from the obtained models. Zhang and Conitzer \shortcite{zhang2019pac} showed that many classical results from Probably Approximately Correct (PAC) learning can be applied to the preference aggregation framework. For a comprehensive study of more related work on AI ethics, interested readers can refer to a recent survey by Yu~\textit{et~al.}~\shortcite{yu2018building}.




The above studies significantly contributed to the emerging area of AI ethics, but unfortunately, none of them is designed to protect voters' privacy in the data analysis process. Moral preferences of voters are sensitive information and should be well protected, since their disclosure can be exploited by malicious entities to have adverse consequences. A recent work~\cite{fredrikson2015model} has shown that even a learned deep neural network can leak sensitive information in the training dataset. Hence, it is critical to incorporate formal privacy protection into models for AI ethics guided by humans. 

To quantify privacy, we use the rigorous notion of differential privacy (DP) \cite{dwork2014algorithmic}. Recently, differential privacy has been widely studied in many areas \cite{abadi2016deep}. For general statistical data release problems, DP can be achieved by the Laplace mechanism \cite{dwork2006calibrating} which injects Laplace noise into the released statistical results.
Furthermore, for parameter estimation solved by an optimization problem, an alternative algorithm to achieve DP is the functional mechanism proposed by~\cite{zhang2012functional}, which perturbs the objective function of the optimization problem rather than the optimized parameters.

As a representative example, this paper adds voter privacy protection into the learned moral preference model of Noothigattu~\textit{et~al.}~\shortcite{noothigattu2018voting}. Incorporating privacy protection into other studies are future directions. Although Jones~\textit{et~al.}~\shortcite{jones2018ai} and Brubaker \shortcite{brubaker2018artificial} emphasized the importance of security and privacy protection in AI ethics, to the best of our knowledge, our paper is the first technical work to formally build privacy protection into the study of AI ethics.

\section{Preliminaries}\label{sec:preliminaries}
Differential privacy (DP) \cite{dwork2006our,dwork2014algorithmic} provides strong guarantees for the privacy of any individual in a query response, regardless of the adversary's prior knowledge. 

\begin{defn}[\textbf{$\epsilon$-Differential Privacy (DP)}~\cite{dwork2014algorithmic,dwork2006our}] \label{defn-DP}A randomized algorithm $Y$ satisfies $\epsilon$-differential privacy, if for any two neighboring datasets $D$ and $D'$ which differ in only one tuple, and for any possible subset of outputs $\mathcal{Y}$ of $Y$, we have
\begin{align} \label{eqn:defn-DP}
\bp{ Y(D) \in \mathcal{Y} } \le e^{\epsilon} \times  \bp{Y(D') \in  \mathcal{Y}},
\end{align}
where $\bp{ \cdot }$ is the probability of an event. $\epsilon$ refers to the privacy parameter. Smaller $\epsilon$ means stronger privacy protection, but less utility as more randomness is introduced into~$Y$.  
\end{defn} 

There are two variants of DP: \mbox{bounded} DP and \mbox{unbounded} DP~\cite{tramer2015differential}. In \mbox{bounded} DP, which we adopt in this paper, two neighboring datasets have the same sizes but different records at only one of all positions. In \mbox{unbounded} DP, the sizes of two neighboring datasets differ by one (i.e., one tuple is in one database, but not in the other). 

The \textbf{Laplace Mechanism}~\cite{dwork2006calibrating} can be used to achieve DP by adding independent Laplace noise to each dimension of the query output. The scale $\lambda$ of the zero-mean Laplace noise $Lap(\lambda)$ is set as $\lambda=\Delta/\epsilon$, where $\Delta$ is the $\ell_1$-norm sensitivity of the query function $Q$, which measures the maximum change of the outputs over neighboring datasets (i.e. \mbox{$\Delta  = \max_{\textrm{neighboring datasets $D,D'$}} \|Q(D) - Q(D')\|_{1}$}). 

\textbf{Problem Formulation.} In our system model, each voter $i\in \{1,2,\cdots,N\}$ owns a dataset of $n$ pairwise comparisons (i.e., $n$ records), denoted by \mbox{$D_i= \{\langle X_1^{(i)},Z_1^{(i)}\rangle, \langle X_2^{(i)},Z_2^{(i)}\rangle, \cdots, \langle X_n^{(i)},Z_n^{(i)}\rangle$ \}.} $X_j^{(i)}$ and $Z_j^{(i)}$ are pairwise alternatives capturing moral dilemmas in a scenario (e.g., $X_j^{(i)}$ means staying the course resulting in killing crossing pedestrians, while $Z_j^{(i)}$ means swerving leading to killing passengers). Any pair $\langle X_j^{(i)},Z_j^{(i)}\rangle$ for $j\in \{1,2,\cdots,n\}$ means that voter $i$ chose $X_j^{(i)}$ over $Z_j^{(i)}$. Each $X_j^{(i)}$ or $Z_j^{(i)}$ is a $d$-dimensional vector such that $X_j^{(i)}[k]$ and $Z_j^{(i)}[k]$ for $k\in\{1,2,\cdots,d\}$ denote the $k$-th dimensional value of $X_j^{(i)}$ and $Z_j^{(i)}$, respectively. The dimension of a scenario represents its features (e.g., the number of young passengers, the number of old pedestrians, the number of pets, etc).



Noothigattu\textit{~et al.~}\shortcite{noothigattu2018voting} adopted the Thurstone--Mosteller process \cite{mosteller2006remarks}, which models the utility as a Gaussian random variable. Let $\bm{\beta}_i\in \mathbb{R}^d$ be the preference parameter of voter $i$, where $ \mathbb{R}$ denotes the set of real numbers. It is assumed~\cite{mosteller2006remarks} that the utilities of alternatives $X_j^{(i)}$ and $Z_j^{(i)}$ follow Gaussian distributions $\mathcal{N}(\bm{\beta}_i^\top X_j^{(i)},\frac{1}{2})$ and $\mathcal{N}(\bm{\beta}_i^\top Z_j^{(i)},\frac{1}{2})$, respectively. Thus, the result of the utility by choosing $X_j^{(i)}$ minus the utility by choosing $Z_j^{(i)}$  follows a Gaussian distribution \mbox{$\mathcal{N}(\bm{\beta}_i^\top(X_j^{(i)}-Z_j^{(i)}),1)$,} so that \\$\bp{\text{a voter $i$ chooses $X_j^{(i)}$ over $Z_j^{(i)}$}} = \Phi (\bm{\beta}_i^\top(X_j^{(i)}-Z_j^{(i)}))$, where $\Phi$ is the cumulative distribution function of the standard normal distribution (i.e. $\Phi(s): = \frac{1}{2}+ \frac{1}{\sqrt{\pi}}\int_{0}^{\frac{s}{\sqrt{2}}}e^{-{t}^2} \hspace{1pt} \emph{d} t $). Then, the maximum likelihood estimation (MLE) method is used to learn the parameter $\bm{\beta}_i$ for each voter $i\in\{1,2,\cdots,N\}$. In particular, the log-likelihood function is defined as follows:
\begin{align}\label{log-func}
    \mathcal{L}(\bm{\beta}_i, D_i) : =\sum_{j=1}^{n}\ln\Phi (\bm{\beta}_i^\top(X_j^{(i)}-Z_j^{(i)})).
\end{align}

Based on Eq.~(\ref{log-func}), MLE is used to estimate the $\bm{\beta}_i$ of each voter \cite{noothigattu2018voting}. However, the optimal parameter $\bm{\beta}_i$ to maximize Eq.~(\ref{log-func}) may not always exist. For example, given $i\in\{1,2,\cdots,N\}$, if $X_j^{(i)}-Z_j^{(i)}$ is positive for each $j\in\{1,2,\cdots,n\}$, suppose there exists $\bm{\beta}_i$ to maximize Eq.~(\ref{log-func}), then each dimension of $\bm{\beta}_i$ is positive. However, a contraction occurs with $\mathcal{L}(2\bm{\beta}_i, D_i) > \mathcal{L}(\bm{\beta}_i, D_i)$. Hence, in the case of positive $X_j^{(i)}-Z_j^{(i)}$ for all $j\in\{1,2,\cdots,n\}$, no  $\bm{\beta}_i$ exists to maximize Eq.~(\ref{log-func}). In order to ensure that (\textit{i}) the optimal parameter can always be found, and (\textit{ii}) bounded sensitivity that will be used in the Laplace mechanism for achieving differential privacy, we introduce a constraint that each voter's parameter has an $\ell_1$-norm at most $B$. Specifically, we define user $i$'s parameter $\overline{\bm{\beta}}_i (D_i)$ by
\begin{align}
\overline{\bm{\beta}}_i (D_i) : =  \argmaxa_{\bm{\beta}_i:~\|\bm{\beta}_i\|_1 \leq B }  \mathcal{L}(\bm{\beta}_i, D_i)  . \label{eq-bar-beta-i}
\end{align}

After learning the parameters $\overline{\bm{\beta}}_i (D_i)$ for all voters, the preference parameter ${\overline{\bm{\beta}}}(D)$ for the whole society is computed by averaging all voters' preference parameters (i.e. ${\overline{\bm{\beta}}}(D):=\frac{1}{N}\sum_{i=1}^N\overline{\bm{\beta}}_i (D_i)$, for  $D:=\{D_1,D_2,\cdots,D_N\}$ denoting the whole dataset of $N$ voters). 
Therefore, the purpose of this paper is to design privacy-preserving algorithms to guarantee voter privacy while learning a society's preference parameter with high accuracy.




\section{Our Solutions}\label{sec:mechanisms}

This section first introduces four privacy protection paradigms and then presents our algorithms to achieve these paradigms.

\subsection{Privacy Modeling}\label{sec:system-model}

Our paper incorporates privacy protection into preference aggregation of Noothigattu~\emph{et~al.}~\shortcite{noothigattu2018voting}. In this setting, each voter's data consist of a number of records, with each record being the voter's preference in a given scenario. In each scenario, no matter how the autonomous vehicle decides to act, someone will get hurt. A voter's preference in each scenario means that the voter prefers a decision made by the autonomous vehicle (e.g., staying the course and killing crossing pedestrians) over the other alternative (e.g., swerving and killing passengers).

Based on the above observations, in our study of privacy protection for crowdsourced data collection in AI ethics, we will consider two variants for the meaning of neighboring datasets: 1) record-neighboring and 2) voter-neighboring datasets, which allow us to achieve \textit{voter-level privacy protection} and \textit{record-level privacy protection}, respectively:
\begin{itemize}
    \item \textbf{Voter-level privacy protection.} Two datasets are voter-neighboring datasets if one can be obtained from the other by changing one voter's records arbitrarily.
    \item \textbf{Record-level privacy protection.} Two datasets are record-neighboring if one can be obtained from the other by changing a \textit{single} record of one voter.
\end{itemize}

For completeness, we also consider the case of a \textit{trusted aggregator} and the case of an \textit{untrusted aggregator}, where privacy is achieved with \textit{centralized perturbation} and \textit{distributed perturbation}, respectively:
\begin{itemize}
    \item \textbf{Centralized perturbation.} When the aggregator is trusted by the voters, the aggregator perturbs the aggregated information (e.g., by adding noise) in a centralized manner to protect privacy. We refer to this as \textit{centralized perturbation}.
    \item \textbf{Distributed perturbation.} When the aggregator is not trusted by the voters, each voter independently perturbs her/his local data (e.g., by adding noise) in a distributed manner for privacy protection. We refer to this as \textit{distributed perturbation}.
\end{itemize}


 

Permutating centralized/distributed perturbation and voter-level/\\record-level privacy protection yields four alternative privacy protection paradigms:
(1) \underline{v}oter-\underline{l}evel privacy protection with \underline{c}entralized \underline{p}erturbation (VLCP); (2) \underline{r}ecord-\underline{l}evel privacy protection with \underline{c}entralized \underline{p}erturbation (RLCP); (3) \underline{v}oter-\underline{l}evel privacy protection with \underline{d}istributed \underline{p}erturbation (VLDP); (4) \underline{r}ecord-\underline{l}evel privacy protection with \underline{d}istributed \underline{p}erturbation (RLDP).
\begin{enumerate}
    \item VLCP: For $\epsilon$-differential privacy with VLCP, the aggregator chooses a universal privacy parameter $\epsilon$ and enforces a randomization algorithm $Y$ such that for any two \textit{voter}-neighboring datasets $D$ and $D'$, and for any possible subset of outputs $\mathcal{Y}$ of $Y$, we have $\bp{ Y(D) \in \mathcal{Y} } \le e^{\epsilon} \times \bp{Y(D') \in \mathcal{Y}}$.
    \item RLCP: For $\epsilon$-differential privacy with RLCP, the aggregator chooses a universal privacy parameter $\epsilon$ and enforces a randomization algorithm $Y$ such that for any two \textit{record}-neighboring datasets $D$ and $D'$, and for any possible subset of outputs $\mathcal{Y}$ of $Y$, we have $\bp{ Y(D) \in \mathcal{Y} } \le e^{\epsilon} \times \bp{Y(D') \in \mathcal{Y}}$.
    \item VLDP: For $\epsilon_i$-differential privacy with VLDP, voter $i$ chooses a privacy parameter $\epsilon_i$ and enforces a randomization algorithm $Y_i$ such that for any two  datasets $D_i$ and $D_i'$ (which are naturally   \textit{voter}-neighboring), and for any possible subset of outputs $\mathcal{Y}_i$ of $Y_i$, we have $\bp{ Y_i(D_i) \in \mathcal{Y}_i } \le e^{\epsilon_i} \times \bp{Y_i(D_i') \in \mathcal{Y}_i}$. Note that VLDP is the same as the notion of $\epsilon$-local differential privacy~\cite{duchi2013local} which has recently received much interest~\cite{wang2019collecting,tang2017privacy,erlingsson2014rappor}.
    \item RLDP: For $\epsilon_i$-differential privacy with VLDP, voter $i$ chooses a privacy parameter $\epsilon_i$ and enforces a randomization algorithm $Y_i$ such that for any two \textit{record}-neighboring datasets $D_i$ and $D_i'$, and for any possible subset of outputs $\mathcal{Y}_i$ of $Y_i$, we have $\bp{ Y_i(D_i) \in \mathcal{Y}_i } \le e^{\epsilon_i} \times \bp{Y_i(D_i') \in \mathcal{Y}_i}$.
\end{enumerate}




In what follows, we further disambiguate the above four privacy protection paradigms, and outline how each of them can be realized.

First, VLCP can be achieved with the data aggregator setting a universal privacy parameter $\epsilon$ and perturbing the averaged moral preference by adding Laplace noise to protect the privacy of each voter's complete data. Second, RLCP can be achieved in the same way as VLCP when using Laplace mechanism. This is because the $\ell_1$ sensitivity of the averaged parameter under RLCP is the same as that under VLCP, which will be proved in Theorem~\ref{thm1-algo-user-level}.
Third, VLDP is a strong privacy protection paradigm and can be achieved by perturbing each voter's moral preference by adding Laplace noise to protect each record in the dataset. Each voter $i$ can choose a personalized privacy parameter $\epsilon_i$ to perturb her/his moral preference accordingly, and report the noisy moral preference to the aggregator. Finally, RLDP can be achieved by the same way as VLDP when using Laplace mechanism. The reason is that the $\ell_1$ sensitivity of each voter's parameter  under RLDP is the same as that under VLDP, which will be proved in Theorem~\ref{thm1-algo-record-level}.

However, achieving RLDP by adding Laplace noise directly leads to limited utility even for weak privacy protection, as illustrated in our experiments in Section~\ref{sec:experiments}. Therefore, to pursuing better utility, we propose to adopt the functional mechanism \cite{zhang2012functional}, which enforces differential privacy by perturbing the object function of the optimization problem, rather than the optimized parameters (i.e., each voter's moral preference). Note that we can't adopt the functional mechanism to estimate the social preference parameter under voter-level privacy protection or centralized perturbation. This is because the social preference parameter is derived by averaging the preference parameters of all voters instead of the solution of the optimization problem. The functional mechanism itself is used for analyzing and solving the optimization problem.

Thus, we will consider the above four privacy protection paradigms to achieve differential privacy for ethical decision making of AI and will propose three algorithms which cover the above four paradigms. Specifically, the three algorithms are outlined as follows.

\begin{itemize}
    \item VLCP/RLCP algorithm via perturbing the average moral preference parameters of all voters by Laplace mechanism. 
    \item VLDP/RLDP algorithm via perturbing the moral preference parameters of each voter by Laplace mechanism.
    \item RLDP algorithm via perturbing the object function of each voter.
\end{itemize}
In the following subsections, we will introduce the proposed algorithms in detail.

\subsection{VLCP/RLCP: Perturbing the Average Moral Preference Parameters}\label{VLCP}

In this section, we propose an algorithm to achieve VLCP and RLCP by perturbing the average moral preference parameters estimated from maximum likelihood estimation.
Each voter $i$ obtains its parameter $\overline{\bm{\beta}}_i(D_i)$ according to Eq.~(\ref{eq-bar-beta-i}), which enforces the $\ell_1$-norm of $\overline{\bm{\beta}}_i(D_i)$ to be at most $B$. Then each voter $i$ sends it to the aggregator, and the aggregator computes the average estimate ${\overline{\bm{\beta}}}(D):=\frac{1}{N}\sum_{i=1}^N\overline{\bm{\beta}}_i (D_i)$, for  $D:=\{D_1,D_2,\cdots,D_N\}$. The $\ell_1$- sensitivity of $\overline{\bm{\beta}}$ with respect to neighboring datasets equals $2B/N$, since the maximal range that $\overline{\bm{\beta}}_i (D_i)$ can change is no greater than $2B$ by the triangle inequality (see Theorem~\ref{thm1-algo-user-level} for specific proofs). As shown in Algorithm~\ref{algo-user-level}, after computing the average parameters of all voters, a random noise vector $R$ will be drawn from the Laplace distribution $[Lap(\frac{2B}{N\epsilon}) ]^d$. Then, the perturbed preference parameter $\bm{\beta}^* (D)={\overline{\bm{\beta}}}(D)+\bm{R}$ is returned as the final social moral preference parameter. Note that we may suppress the argument $D_i$ in $\overline{\bm{\beta}}_i (D_i)$ and the argument $D$ in ${\overline{\bm{\beta}}}(D),\bm{\beta}^* (D)$ for simplicity.

\begin{algorithm}[ht]
\small
\setstretch{1} 
 \caption{\small{VLCP/RLCP: Perturbing the Average Moral Preference Parameters by Adding Laplace Noise}} \label{algo-user-level}
  \KwIn{voters pairwise comparison data $D=\{D_1,D_2,\cdots,D_N\}$, universal privacy parameter $\epsilon$, norm bound $B$} 
  \KwOut{noisy parameter $\bm{\beta}^*$}
    \For{each voter $i$ with data $D_i=\{\langle X_1^{(i)},Z_1^{(i)}\rangle,\cdots,\langle X_n^{(i)},Z_n^{(i)}\rangle\}$} 
    {Apply MLE to estimate $\overline{\bm{\beta}}_i$ of voter $i$ with log-likelihood function $\mathcal{L}(\bm{\beta}_i)=\sum_{j=1}^{n}\ln\Phi (\bm{\beta}_i^\top(X_j^{(i)}-Z_j^{(i)}))$, subject to $\|\overline{\bm{\beta}}_i\|_1\leq B$\;  voter $i$ sends $\overline{\bm{\beta}}_i$ to the aggregator\;}  
    The aggregator computes the average estimate: $\overline{\bm{\beta}}=\frac{1}{N}\sum_{i=1}^N\overline{\bm{\beta}}_i$\;
    The aggregator draws a random Laplace noise vector $\bm{R}\sim [ Lap(\frac{2B}{N\epsilon}) ]^d$\;
    \textbf{Return} $\bm{\beta}^*=\overline{\bm{\beta}}+\bm{R}$\; 
\end{algorithm}

\begin{thm}[Privacy of Algorithm~\ref{algo-user-level}] \label{thm1-algo-user-level}
Algorithm~\ref{algo-user-level} satisfies $\epsilon$-differential privacy for both VLCP and RLCP.
\end{thm}

\begin{proof}
First, we prove that  the sensitivities $\Delta_{VLCP} $ and $\Delta_{RLCP}$ with respect to voter-neighboring datasets and record-neighboring datasets are both $2B/N$ under centralized perturbation, where
\begin{small}
\begin{align}    
\Delta_{VLCP}  &:= \max_{\textrm{$D$ and $D'$ differing in a single voter's records}} \|\overline{\bm{\beta}}(D) - \overline{\bm{\beta}}(D')\|_1,\\
\Delta_{RLCP}  &:= \max_{\textrm{$D$ and $D'$ differing in a single voter's one record}} \|\overline{\bm{\beta}}(D) - \overline{\bm{\beta}}(D')\|_1 .
\end{align}
\end{small}

Let $D=\{D_1,\cdots,D_i,\cdots,D_N\}$ and $D'=\{D_1,\cdots,D'_i,\cdots,D_N\}$ be record-neighboring datasets such that $D$ and $D'$ differ in only a single voter $i$'s one record, i.e., $D_i$ and $D'_i$. We assume {$D_i := \{\langle X_1^{(i)},Z_1^{(i)}\rangle, \langle X_2^{(i)},Z_2^{(i)}\rangle, \cdots, \langle X_n^{(i)},Z_n^{(i)}\rangle$ \}}, where each dimension of $X_1^{(i)} - Z_1^{(i)} $ is $1$, and each dimension of $X_j^{(i)} - Z_j^{(i)} $ for $j \in \{2, \ldots, n\}$ is close to $0$, and {$D_i' := \{\langle {X_1^{(i)}}',{Z_1^{(i)}}'\rangle, \langle X_2^{(i)},Z_2^{(i)}\rangle, \cdots, \langle X_n^{(i)},Z_n^{(i)}\rangle$ \}}, where each dimension of $X_1^{(i)} - Z_1^{(i)} $ is $-1$, and each dimension of $X_j^{(i)} - Z_j^{(i)} $ for $j \in \{2, \ldots, n\}$ is close to $0$. Therefore, based on Eq.~(\ref{eq-bar-beta-i}), we know that each of the $d$ dimensions of $\overline{\bm{\beta}}_i (D_i)$ and $\overline{\bm{\beta}}_i (D_i')$ are close to $\frac{B}{d}$ and $-\frac{B}{d}$, respectively. Given ${\overline{\bm{\beta}}}(D):=\frac{1}{N}\sum_{i=1}^N\overline{\bm{\beta}}_i (D_i)$ and ${\overline{\bm{\beta}}}(D'):=\frac{1}{N}\overline{\bm{\beta}}_i (D_i')+\frac{1}{N}\sum_{j\in\{1,2,\ldots,N\}\setminus\{i\}} \overline{\bm{\beta}}_j (D_j)$, the $\ell_1$-norm difference between ${\overline{\bm{\beta}}}(D)$ and ${\overline{\bm{\beta}}}(D')$ equals $\frac{1}{N}\overline{\bm{\beta}}_i (D_i)-\frac{1}{N}\overline{\bm{\beta}}_i (D_i')$, which can be made to be arbitrarily close to $\frac{2B}{N}$.  Hence, the $\ell_1$-norm sensitivity of ${\overline{\bm{\beta}}} $ with respect to record-neighboring datasets is $\frac{2B}{N}$.

In a way similar to the above argument, the $\ell_1$-norm sensitivity of ${\overline{\bm{\beta}}} $ with respect to voter-neighboring datasets is $\frac{2B}{N}$. For voter-neighboring datasets $D=\{D_1,\cdots,D_i,\cdots,D_N\}$ and $D'=\{D_1,\cdots,D'_i,\cdots,D_N\}$, voter $i$'s dataset variants $D_i$ and $D_i'$ can differ arbitrarily.


Thus, we obtain $\Delta _{RLCP} = \Delta _{VLCP}  = 2B/N$. Recall that Algorithm~\ref{algo-user-level} considers the centralized perturbation paradigm and adds noises drawn from $[ Lap(\frac{2B}{N\epsilon}) ]^d$ to the average parameters of all voters. From Laplace mechanism~\cite{dwork2006calibrating}, it can be proved that Algorithm~\ref{algo-user-level} satisfies $\epsilon$-differential privacy for both VLCP and RLCP.
\end{proof}


We now analyze the utility of Algorithm~\ref{algo-user-level}. At the end of Algorithm~\ref{algo-user-level}, the aggregator obtains the parameter vector of moral preference $\bm{\beta^*}$, which can be understood as a noisy version of the true result $\overline{\bm{\beta}}$. We consider the utility of Algorithm~\ref{algo-user-level} by analyzing the probability that the $p$-norm of the estimation error $\bm{\beta^*}  - \overline{\bm{\beta}}$ is no greater than a given quantity $\alpha$. To this end, we note that  $\bm{\beta^*}  - \overline{\bm{\beta}}$ follows a $d$-dimensional multivariate Laplace distribution, with each dimension being  an independent zero-mean Laplace random variable with scale $\frac{2B}{N\epsilon}$. With $h(\bm{\mu})$ denoting the probability density function of random variable $\bm{\beta^*}  - \overline{\bm{\beta}}$ being a given $\bm{\mu}$, the expression of $h(\bm{\mu})$ involves Bessel function.  
Hence, for general $p$-norm, it is difficult to compute $\bp{\left \| \bm{\beta^*} - \overline{\bm{\beta}} \right \|_p\leq \alpha} = \int_{\bm{\mu}:\,\left \|\bm{\mu} \right \|_p\leq \alpha} h(\bm{\mu}) \, \text{d}\bm{\mu} $. Below we consider the special case of $\infty$-norm following~\cite{dwork2014algorithmic}, where we use the union bound to present the utility result of Algorithm~\ref{algo-user-level} in the following Theorem~\ref{thm-algo-user-level-utility}.
 

\begin{thm}[Utility of Algorithm~\ref{algo-user-level}] \label{thm-algo-user-level-utility}
For any $\gamma \in (0,1)$, Algorithm~\ref{algo-user-level} ensures $\left\| \bm{\beta^*}  - \overline{\bm{\beta}} \right\|_\infty\leq \alpha$ with probability at least $1-\gamma\,$ for $\alpha := \frac{2B}{N\epsilon} \cdot \ln(d/\gamma)$.
\end{thm}

\begin{proof}
With $\bm{R}:=\bm{\beta^*} - \overline{\bm{\beta}}$, we know that $\bm{R}$ follows the probability distribution $ [ Lap(\frac{2B}{N\epsilon}) ]^d$. Then denoting the $d$ dimensions of $\bm{R}$ by $R[1], R[2], \ldots, R[d]$, we have
\begin{align}
     & \bp{\left \| \bm{\beta^*} - \overline{\bm{\beta}} \right \|_\infty\geq \alpha}
  = \bp{\left \| \bm{R} \right \|_\infty\geq \alpha} \nonumber\\
    &= \bp{ \left(\max_{k\in\{1,2,\ldots, d\}} | R[k] | \right) \hspace{-2pt}\geq\hspace{-2pt} \alpha  } \hspace{-2pt}\leq\hspace{-2pt} \sum_{k\in\{1,2,\ldots, d\}} \bp{ | R[k] | \hspace{-2pt}\geq\hspace{-2pt} \alpha  } \nonumber\\
    & =  \textstyle{d\cdot \bp{|Lap(\frac{2B}{N\epsilon})| \geq \alpha} } =  \gamma \text{ for $\alpha := \frac{2B}{N\epsilon} \cdot \ln(d/\gamma)$},\nonumber
\end{align}
where the step of ``$\leq$'' uses the union bound.
\end{proof}

Theorem~\ref{thm-algo-user-level-utility} shows that the utility   of Algorithm~\ref{algo-user-level} decreases as $d$ increases. This is confirmed by Figure~\ref{syn-vary-d} for experimental results.

From Theorem~\ref{thm1-algo-user-level}, both VLCP and RLCP assume a universal privacy parameter for all voters; i.e., the same privacy protection level for all voters. In the following, we focus on proposing algorithms with distributed perturbation so that voters can choose personalized privacy parameters to achieve different privacy protection levels.

\subsection{VLDP/RLDP: Perturbing the Moral Preference Parameter of Each Voter}\label{sec:record-level-lap}

This section introduces an algorithm to achieve VLDP and RLDP by perturbing the moral preference parameter of each voter under Laplace mechanism. Algorithm~\ref{algo-record-level} shows the pseudo-code of perturbing the preference parameter of each voter $i$ with a personalized privacy parameter $\epsilon _i$.
Each voter $i$ obtains its parameter $\overline{\bm{\beta}}_i(D_i)$ according to Eq.~(\ref{eq-bar-beta-i}), which enforces the $\ell_1$-norm of $\overline{\bm{\beta}}_i(D_i)$ to be at most $B$.
Then, the sensitivity of each voter's parameter $\bm{\beta}_i$ will $2B$ since the maximal changing range of $\bm{\beta}_i$ is no greater than $2B$ with respect to neighboring datasets (see Theorem~\ref{thm1-algo-record-level} for the specific proofs). Then, the parameter of voter $i$ will be perturbed as $\overline{\bm{\beta}}_i^*(D_i)=\overline{\bm{\beta}}_i(D_i)+ \bm{R}$, where $\bm{R}$ is a random Laplace noise vector drawn from $[Lap(2B/\epsilon _i)]^d$.

After obtaining the noisy parameter of each voter $i$ in a distributed way, the final social moral preference parameter can be computed by averaging all voters' parameters, that is $\bm{\beta}^*(D)=\frac{1}{N}\sum_{i=1}^N\overline{\bm{\beta}}_i^*(D_i)$ for $D:=\{D_1,\cdots,D_i,\cdots,D_N\}$. Note that we may suppress the argument $D_i$ in $\overline{\bm{\beta}}_i (D_i),\overline{\bm{\beta}}_i^*(D_i)$ and the argument $D$ in $\bm{\beta}^* (D)$ for simplicity.

\begin{algorithm}[t]
\small
\setstretch{1} 
 \caption{\small{VLDP/RLDP: Perturbing the Preference Parameter of Each Voter by Adding Laplace Noise}} \label{algo-record-level}
  \KwIn{\mbox{voter $i$'s dataset $D_i = \{\langle X_j,Z_j\rangle |j\in\{1,2,\cdots,n\}\}$}, personalized privacy parameter $\epsilon_i$, norm bound $B$}
  \KwOut{noisy parameter $\overline{\bm{\beta}}_i^*$ of voter $i$}
     Apply MLE to estimate $\bm{\beta}_i$ of voter $i$ with log-likelihood function $\mathcal{L}(\bm{\beta}_i)=\sum_{j=1}^{n}\ln\Phi (\bm{\beta}_i^\top(X_j-Z_j))$, subject to $\|\overline{\bm{\beta}}_i\|_1\leq B$\;
    Draw a random Laplace noise vector $\bm{R} \sim [Lap(2B/\epsilon _i)]^d$\;
    \textbf{return} $\overline{\bm{\beta}}_i^*=\overline{\bm{\beta}}_i + \bm{R}$\;
\end{algorithm}

\begin{thm}[Privacy of Algorithm~\ref{algo-record-level}] \label{thm1-algo-record-level}
For each voter $i$, Algorithm~\ref{algo-record-level} satisfies $\epsilon _i$-differential privacy for both VLDP and RLDP.
\end{thm}

\begin{proof}

First, we prove that for voter $i$, the sensitivities $\Delta_{VLDP}^{(i)}$ and $\Delta_{RLDP}^{(i)}$ with respect to voter-neighboring and record-neighboring datasets are both $2B$ under distributed perturbation, where
\begin{small}
\begin{align}  
\Delta_{VLDP}^{(i)}  &:= \max_{D_i,\,D_i'} \|\overline{\bm{\beta}}_i(D_i) - \overline{\bm{\beta}}_i(D_i')\|_1,\\
\Delta_{RLDP}^{(i)}  &:= \max_{\textrm{$D_i$ and $D_i'$ differing in one record}} \|\overline{\bm{\beta}}_i(D_i) - \overline{\bm{\beta}}_i(D_i')\|_1.
\end{align}
\end{small}


For voter $i$, let $D_i$ and $D_i'$ be two record-neighboring datasets. Specifically, for {$D_i := \{\langle X_1^{(i)},Z_1^{(i)}\rangle, \langle X_2^{(i)},Z_2^{(i)}\rangle, \cdots, \langle X_n^{(i)},Z_n^{(i)}\rangle$ \}}, where each dimension of $X_1^{(i)} - Z_1^{(i)} $ is $1$, and each dimension of $X_j^{(i)} - Z_j^{(i)} $ for $j \in \{2, \ldots, n\}$ is close to $0$, we know from Eq.~(\ref{eq-bar-beta-i}) that each of the $d$ dimensions of $\overline{\bm{\beta}}_i (D_i)$ is close to $\frac{B}{d}$. For {$D_i' := \{\langle {X_1^{(i)}}',{Z_1^{(i)}}'\rangle, \langle X_2^{(i)},Z_2^{(i)}\rangle, \cdots, \langle X_n^{(i)},Z_n^{(i)}\rangle$ \}}, where each dimension of $X_1^{(i)} - Z_1^{(i)} $ is $-1$, and each dimension of $X_j^{(i)} - Z_j^{(i)} $ for $j \in \{2, \ldots, n\}$ is close to $0$, we know from Eq.~(\ref{eq-bar-beta-i}) that each of the $d$ dimensions of $\overline{\bm{\beta}}_i (D_i')$ is close to $-\frac{B}{d}$. Hence, the $\ell_1$-norm difference between $\overline{\bm{\beta}}_i (D_i)$ and $\overline{\bm{\beta}}_i (D_i')$ can be made to be arbitrarily close to $2B$. Hence, for voter $i$, the $\ell_1$-norm sensitivity of $\overline{\bm{\beta}}_i$ with respect to record-neighboring datasets is $2B$.

In a way similar to the above argument, the $\ell_1$-norm sensitivity of $\overline{\bm{\beta}}_i$ with respect to voter-neighboring datasets is $2B$. Note that voter $i$'s voter-neighboring datasets $D_i$ and $D_i'$ can differ arbitrarily. 


Therefore, we obtain   $\Delta _{VLDP}^{(i)} = \Delta _{RLDP}^{(i)}  = 2B$. 
Since Algorithm~\ref{algo-record-level} adds noises drawn from $[Lap(2B/\epsilon _i)]^d$ to the parameter of each voter $i$, thus it satisfies $\epsilon _i$-differential privacy for both VLDP and RLDP based on Laplace mechanism~\cite{dwork2006calibrating}.
\end{proof}

\begin{thm}[Utility of Algorithm~\ref{algo-record-level}] \label{thm-algo-record-level-utility}
For any $\gamma \in (0,1)$, Algorithm~\ref{algo-record-level} ensures $\left\| \overline{\bm{\beta}}_i^*  - \overline{\bm{\beta}}_i \right\|_\infty\leq \alpha$ with probability at least $1-\gamma\,$ for $\alpha := \frac{2B}{\epsilon _i} \cdot \ln(d/\gamma)$.
\end{thm}

\begin{proof}
With $\bm{R}:=\overline{\bm{\beta}}_i^*  - \overline{\bm{\beta}}_i$, we know that $\bm{R}$ follows the probability distribution $[Lap(2B/\epsilon _i)]^d$. Then, by referring to the proof of Theorem~\ref{thm-algo-user-level-utility}, we can easily obtain the conclusion of Theorem~\ref{thm-algo-record-level-utility}. We omit the specific proof due to the space limitation.
\end{proof}

Our experiments in Section~\ref{sec:experiments} will show that achieving VLDP/RLDP by adding Laplace noise to each voter's parameter will lead to a limited utility. Recall that VLDP (voter-level privacy protection with distributed perturbation) itself is the strongest notion among four paradigms focused in this paper. Thus, it accordingly has the worst data utility. Clearly, we can know that the definition of RLDP is relatively weaker than VLDP. However, RLDP holds the same low data utility as the VLDP when using Laplace mechanism. Therefore, in the next section, we focus on proposing a novel algorithm which can achieve RLDP while ensuring a higher data utility.

\subsection{RLDP: Perturbing the Object Function of Each Voter}\label{sec:record-level-func}


To further enhance the data utility, an alternative approach is to perturb the objective function of each voter when conducting MLE optimization. This is the functional mechanism of \cite{zhang2012functional}. For comparison, we will show in Section~\ref{sec:experiments} (Figures~\ref{syn-compare123} and~\ref{real}(d)) that the functional mechanism can provide better utility than adding Laplace noise directly for record-level differential privacy.

In view of Eq.~(\ref{log-func}), we define $V_{j}^{(i)}: = X_{j}^{(i)} - Z_{j}^{(i)}$.
We split the log-likelihood function for voter $i\in \{1,2,\cdots,N\}$ as
\begin{align}\label{user-func}
\textstyle{\mathcal{L}(\bm{\beta}_i, V_j^{(i)}|_{j \in \{1,2,\ldots,n\}})=\sum_{j=1}^{n}\ln\Phi (\bm{\beta}_i^\top V_j^{(i)} ).} 
\end{align}
Later, we will focus on one arbitrary voter $i$, and omit the notation $i$ in the analysis for simplicity. Since the final parameter $\bm{\beta}$ is just the average of each $\bm{\beta}_i$, we can process each $\bm{\beta}_i$ one by one. 
 We then split the objective function (\ref{user-func}) for each record $j\in \{1,2,\cdots,n\}$ as
\begin{align}\label{obj-func-2}
\textstyle{f(V_j,\bm{\beta})=\ln\Phi (\bm{\beta}^\top V_j).}  
\end{align}
Therefore, the optimal model parameter $\bm{\beta}^o$ for each voter is:
\begin{align}\label{obj-func}
 \textstyle{\bm{\beta}^o = \argmaxa\limits_{\bm{\beta}} \sum_{j=1}^{n}f(V_j,\bm{\beta}).}    
\end{align}

We denote a $d$-dimensional vector $\bm{\beta}$ by $(\beta[1],\cdots,\beta[d])$. By the Stone-Weierstrass Theorem \cite{rudin1964principles}, any continuous and differentiable $f(V_j,\bm{\beta})$ can always be written as a (potentially infinite) polynomial of $(\beta[1],\cdots,\beta[d])$. For some $W\in[0,\infty]$, we have
\begin{align}\label{poly}
 \textstyle{f(V_j,\bm{\beta})=\sum_{w=0}^{W}\sum_{\phi \in \Psi_w} \lambda_{\phi_{,V_j}}\phi(\bm{\beta}),}    
\end{align}
where $\lambda_{\phi_{,V_j}} \in \mathbb{R}$ is the coefficient of $\phi(\bm{\beta})$ in the polynomial. $\phi(\bm{\beta})$ denotes the product of $\beta[1],\cdots,\beta[d]$ for some $c_1,\cdots,c_d\in \mathbb{N}$ (the set of non-negative integers). $\Psi_w(w\in \mathbb{N})$ contains the set of all products of $\beta[1],\cdots,\beta[d]$ with degree $w$, that is\\
$\Psi_w=\left\{ (\beta[1])^{c_1}(\beta[2])^{c_2}\cdots(\beta[d])^{c_d} \mid \sum_{l=1}^d c_l=w \right\}$.



Thus, we will express the objective function (\ref{obj-func-2}) for each user $i$ and each record $j$ as a polynomial like Eq.~(\ref{poly}). 
Let $f_1$ and $g_1$ be two functions defined as follows:
\begin{align}
    g_1(V_j,\bm{\beta})=\bm{\beta}^\top V_j \text{ and } f_1(z)=\ln\Phi(z).
\end{align}
Then we have $f(V_j,\bm{\beta})=f_1(g_1(V_j,\bm{\beta}))$. 

Using the \textit{Taylor expansion} at $0$, we write $f(V_j,\bm{\beta})$ as  $\tilde{f}_D(\bm{\beta})$ below: 
\begin{align}\label{eqn-4}
  \textstyle{\tilde{f}_D(\bm{\beta})=\sum_{j=1}^{n}\sum_{k=0}^{\infty}\frac{f_1^{(k)}(0)}{k!}(\bm{\beta}^\top V_j)^k. }   
\end{align}


Then we adopt an approximation approach to reduce the degree of the summation. In particular, we only use the value of $f_1^{(k)}(0)$ for $k=0,1,2$, where $f_1^{(k)}$ denotes the $k$-order derivative. Specifically, we can compute $f_1^{(0)}(0)=\ln\frac{1}{2}$. From $f_1^{(1)}(z)=\frac{{\Phi}^{'}(z)}{\Phi(z)} $, we get $f_1^{(1)}(0)=\sqrt{2/\pi}$. From $f_1^{(2)}(z)=\frac{{\Phi}^{''}(z)\cdot \Phi(z) - {\Phi}^{'}(z)\cdot {\Phi}^{'}(z)}{\Phi^2(z)} $,   we have $f_1^{(2)}(0)=-2/\pi$. Thus, we approximate $\tilde{f}_D(\bm{\beta})$ by 
\begin{align}\label{eqn-5}
 \textstyle{\hat{f}_D(\bm{\beta}):=\sum_{j=1}^{n}\sum_{k=0}^{2}\frac{f_1^{(k)}(0)}{k!}(\bm{\beta}^\top V_j)^k.}    
\end{align}


For each $i\in\{1,2,\cdots,N\}$, $j\in\{1,2,\cdots,n\}$, we preprocess the data to make $\|X_j^{(i)}\|_2\leq 1/2 $ and $\|Z_j^{(i)}\|_2\leq 1/2 $. Then $\|V_j^{(i)}\|_2\leq 1 $ by triangle inequality. For simplicity, we suppress both $i$ and $j$, and write $V \in D$ to mean $V \in \{V' : X' - Z' | \langle X',Z'\rangle \in D \}$. Given $\|V \|_2 = \sqrt{\sum_{k=1}^{d}(V[k])^2}\leq 1$,   we have $ \sum_{k=1}^{d}V[k]\leq \sqrt{d}$ by Cauchy-Schwarz inequality, where $V[k]$ denotes the $k$-th dimension of the vector $V$. Thus, from Lemma~1 in \cite{zhang2012functional} and Eq.~(\ref{eqn-5}), we can compute the $\ell_1$-sensitivity $\Delta$ of the coefficient vector for $\beta[1],\cdots,\beta[d]$ in Eq.~(\ref{eqn-4}) with respect to record-neighboring datasets as
\begin{align} \label{sen-func}
    \Delta &= \textstyle{  2\cdot \underset{V \in D}\max\sum_{w=1}^{W}\sum_{\phi \in \Psi_w} \left \| \lambda_{\phi_{,V}}\right \|_1} \nonumber\\
    &= \textstyle{ 2\cdot \underset{V \in D}\max\left ( \frac{f_1^{1}(0)}{1!}\sum_{k=1}^{d}V[k] + \frac{f_1^{2}(0)}{2!}\sum_{1\leq k,l\leq d} V[k]V[l] \right)}\nonumber\\
    &\leq \textstyle{  2\left(\sqrt{2d/\pi}+d/\pi\right)}, \nonumber
    \text{which we denote by $\Delta_{\text{upper}}$ below.}
\end{align}

\begin{algorithm}[t]
\small
\setstretch{1} 
 \caption{\small{RLDP: Perturbing the Object Function of Each Voter}} \label{algo-record-level-func}
  \KwIn{\mbox{voter $i$'s dataset $D_i = \{\langle X_j,Z_j\rangle |j\in\{1,2,\cdots,n\}\}$} with $d$-dimensional vector $V_j:=X_j - Z_j$, objective function $f_{D_i}(\bm{\beta}_i)$, personalized privacy parameter $\epsilon_i$} 
  \KwOut{noisy parameter $\bm{\beta}_i^*$ of voter $i$}
  
    
    Decompose cost function as $f(V_j,\bm{\beta}_i)=f_1(g_1(V_j,\bm{\beta}_i))$\;
    Build an approximate objective function $\hat{f}_{D_i}(\bm{\beta}_i)$, such that $\hat{f}_{D_i}(\bm{\beta}_i)=\sum_{j=1}^{n}\sum_{k=0}^{2}\frac{f_1^{(k)}(0)}{k!}(g_1(V_j,\bm{\beta}_i))^k$\;
    Set $\Delta_{\text{upper}} = 2\sqrt{2d/\pi}+2d/\pi$\;
    \For{each $0\leq w \leq 2$}{
    \For{each $\phi \in \Psi_w$}
    {Compute $\lambda_\phi=\sum_{V_j\in {D_i}}\lambda_{\phi_{,V_j}}+Lap(\Delta_{\text{upper}}/\epsilon_i)$\;}}
    Let $f_{D_i}^*(\bm{\beta}_i)=\sum_{w=0}^{2}\sum_{\phi \in \Psi_w} \lambda_{\phi} \phi(\bm{\beta}_i)$\;
    $\bm{\beta}_i^*=\argmaxa\limits_{\bm{\beta}_i}f_{D_i}^*(\bm{\beta}_i)$\;
    \textbf{Return} $\bm{\beta}_i^*$\; 
\end{algorithm}

After computing the sensitivity's upper bound $\Delta_{\text{upper}}$ as above, we inject Laplace noise with scale $\Delta_{\text{upper}}/\epsilon_i$ to the coefficients of the objective function and then solve the noisy objective function to estimate the moral parameters $\bm{\beta}_i^*$ of each voter $i$. Algorithm~\ref{algo-record-level-func} shows the pseudo-code of our algorithm. From the above analysis, we have a Theorem~\ref{thm1-algo-record-level-func} which follows the above analysis.


\begin{thm}[Privacy of Algorithm~\ref{algo-record-level-func}] \label{thm1-algo-record-level-func}
Algorithm~\ref{algo-record-level-func} satisfies $\epsilon_i$-differential privacy for RLDP.
\end{thm}





The utility of Algorithm~\ref{algo-record-level-func} is bounded by a small quantity, which can be proved by following the utility analysis of the functional mechanism in \cite{zhang2012functional}. We omit the details due to space limitation.

\begin{figure*}[h]
\vspace{-3mm}
  \begin{minipage}{.48\textwidth}
    \centering
  \begin{tabular}{cc} 
     \includegraphics[height=2.8cm]{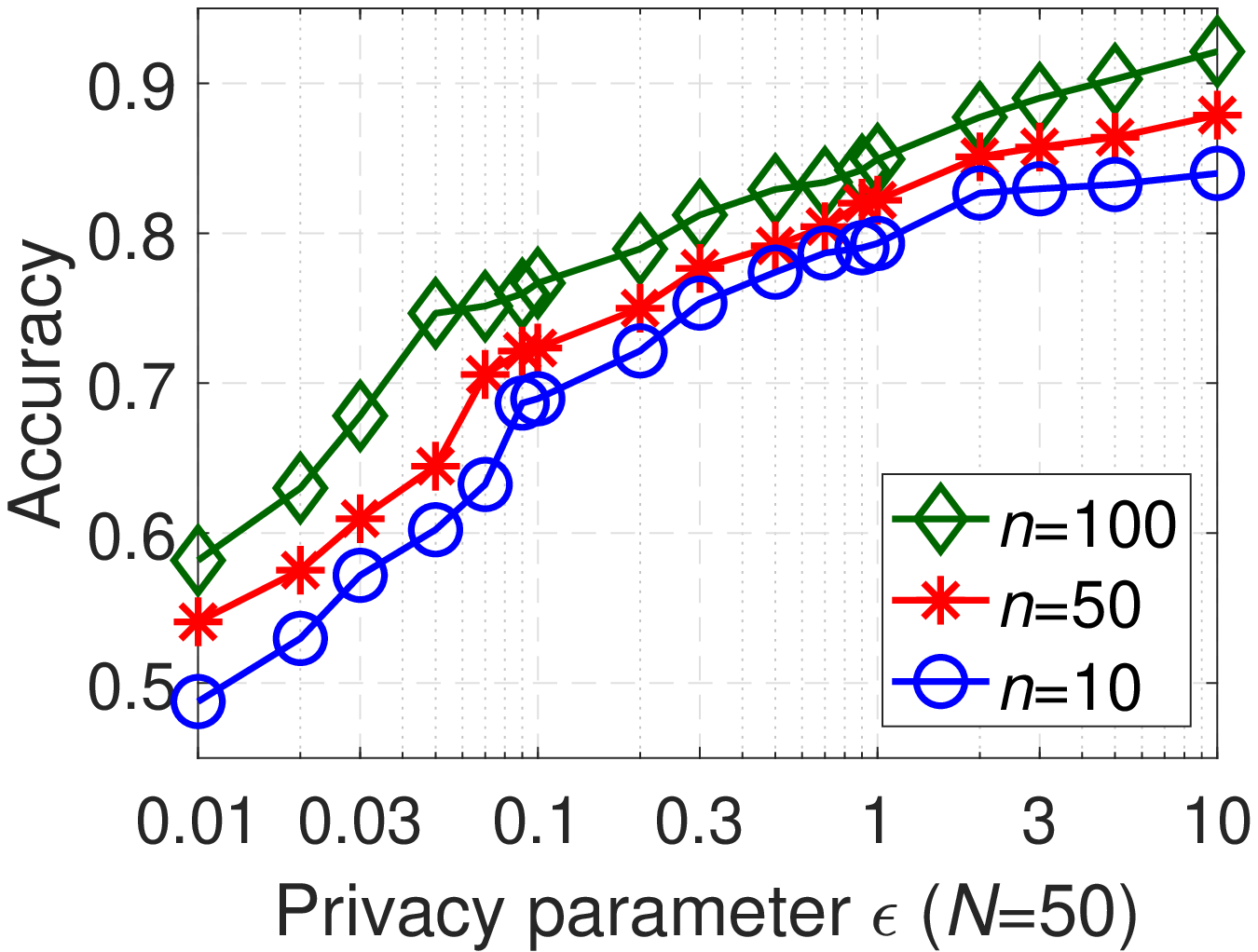} &
    \hspace{-3mm}\includegraphics[height=2.8cm]{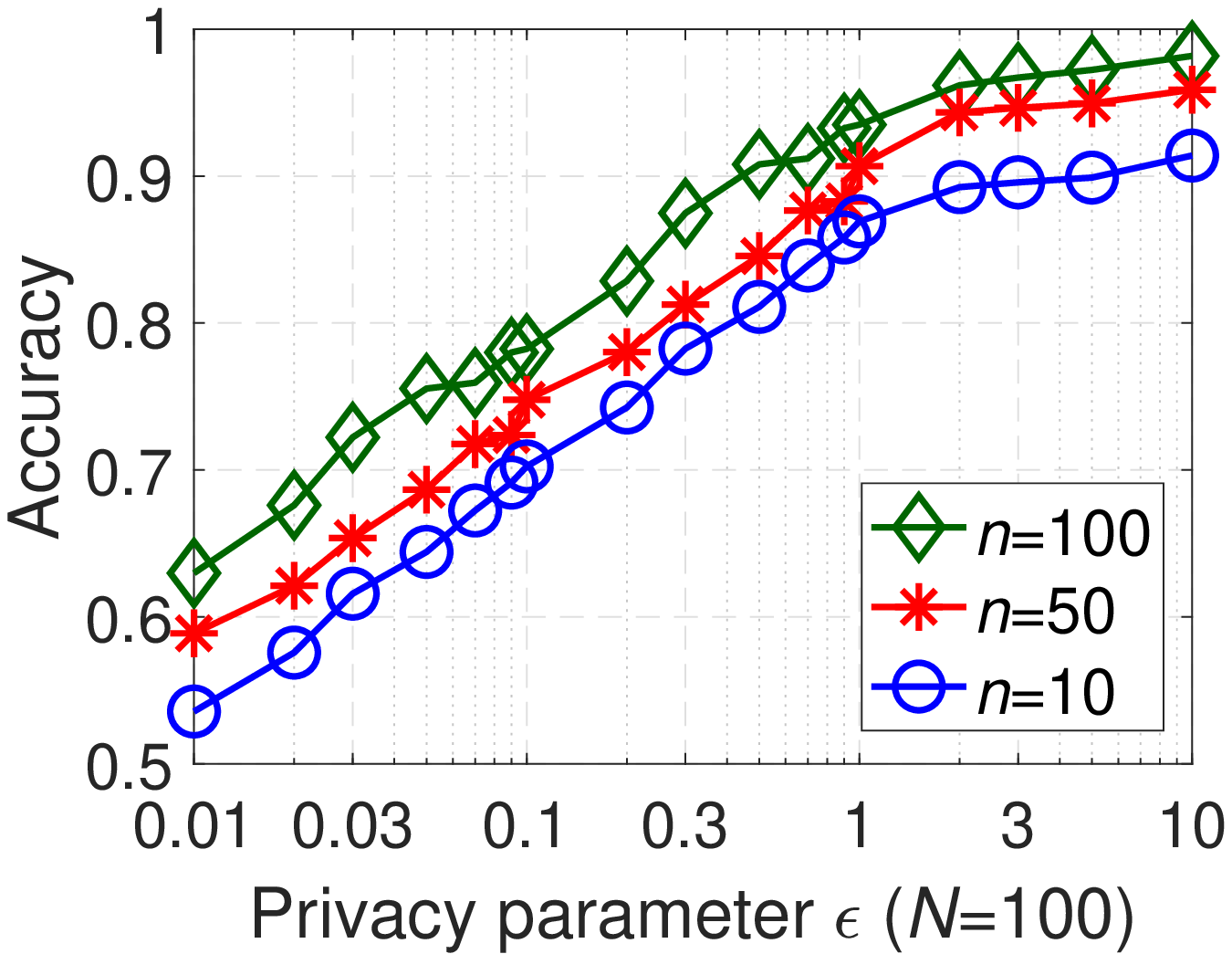}
   \vspace{-1mm} \\
   \footnotesize(a) Number of voters $N=50$ & 
   \hspace{-2mm}\footnotesize(b) Number of voters $N=100$
   \end{tabular} 
 	\vspace{-4mm}
    \caption{\small{Accuracy of Algorithm~\ref{algo-user-level} (VLCP/RLCP under Laplace mechanism) \textit{vs.} privacy parameter $\epsilon$ on synthetic dataset.}}
    \vspace{-3mm} 
	\label{syn-UL}
	\end{minipage}
	~~~~~
\hspace{5mm}
  \begin{minipage}{.48\textwidth}
  \centering
  \begin{tabular}{cc} 
     \includegraphics[height=2.8cm]{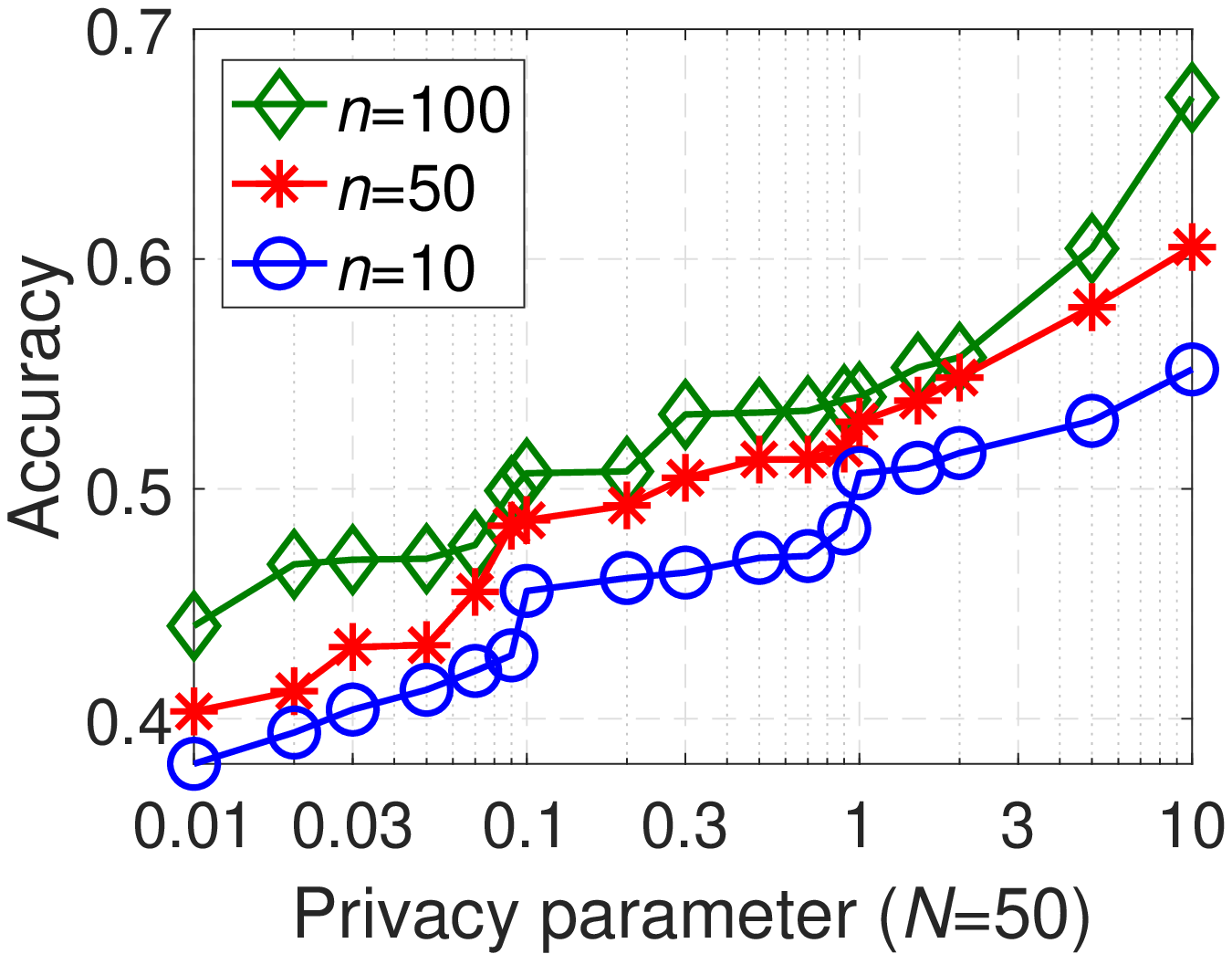} & \hspace{-3mm}\includegraphics[height=2.8cm]{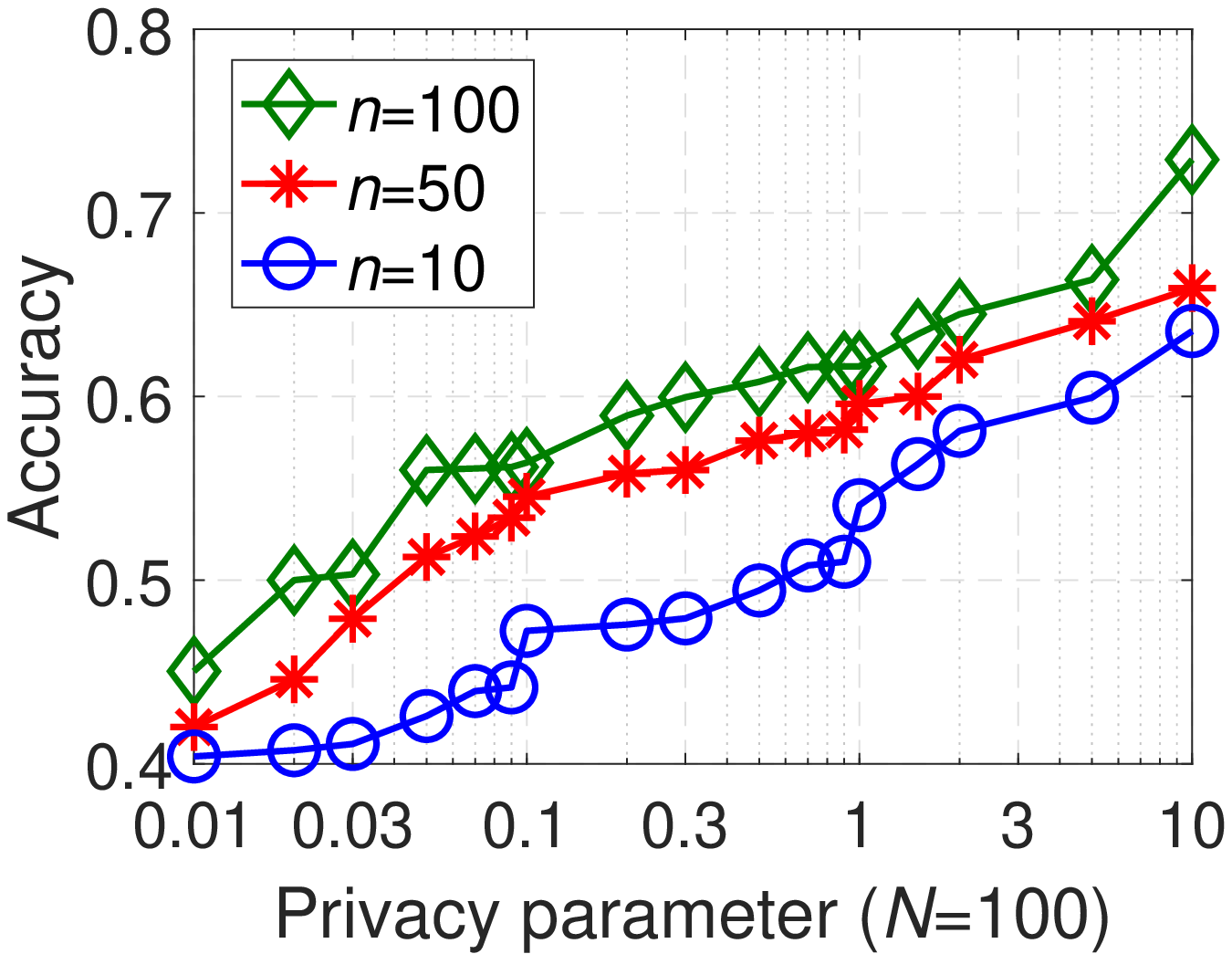} \vspace{-1mm} \\
   \footnotesize(a) Number of voters $N=50$ & \hspace{-2mm}\footnotesize(b) Number of voters $N=100$
   \end{tabular} 
 	\vspace{-4mm}
    \caption{\small{Accuracy of Algorithm~\ref{algo-record-level} (VLDP/RLDP under Laplace mechanism) \textit{vs.} privacy parameter $\epsilon$ on synthetic dataset.}}
    \vspace{-3mm} 
	\label{syn-EL}
	\end{minipage}
\end{figure*}

\begin{figure*}
\begin{minipage}{.48\textwidth}
\centering
  \begin{tabular}{cc} 
     \includegraphics[height=2.8cm]{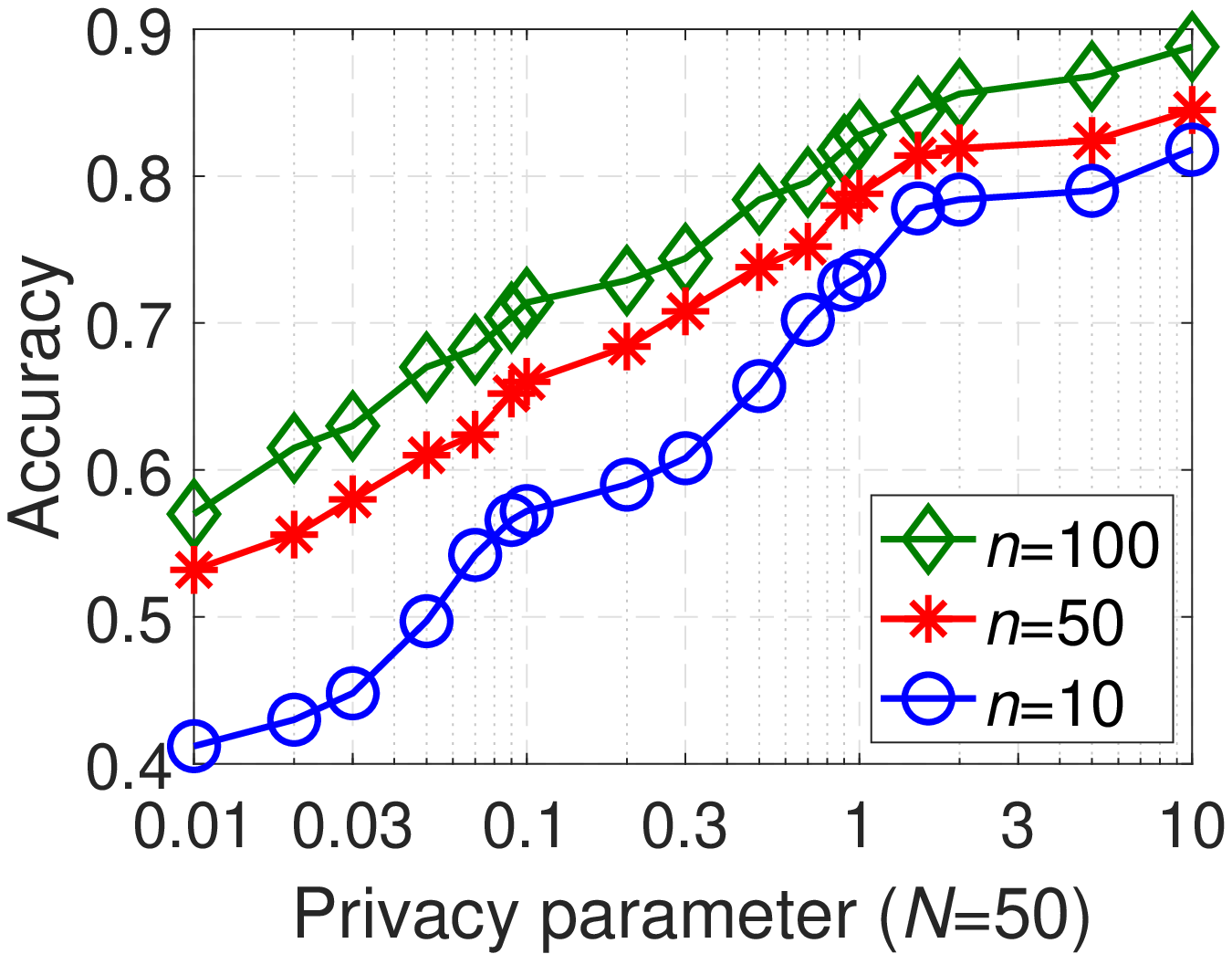} & \hspace{-3mm}\includegraphics[height=2.8cm]{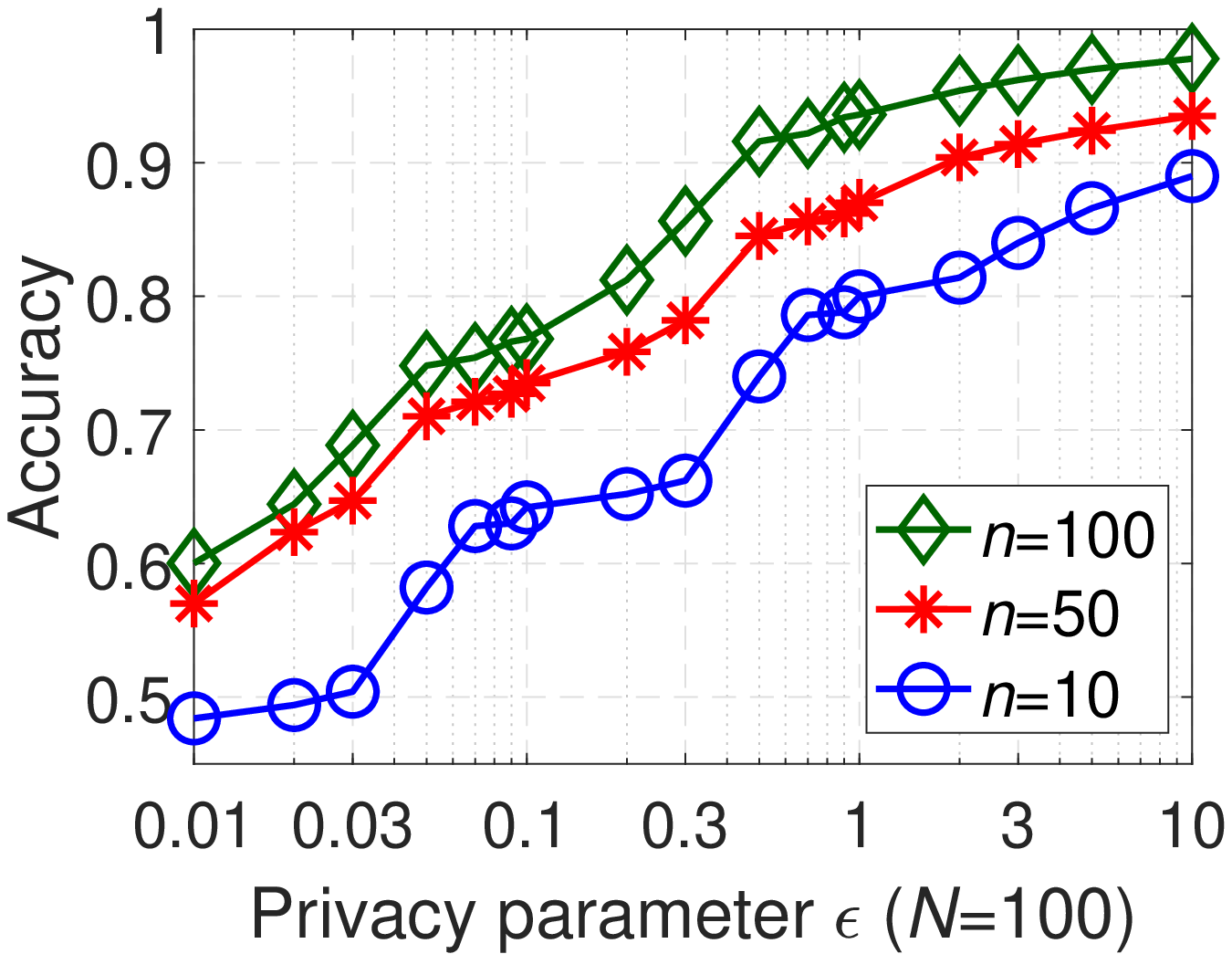} \vspace{-1mm} \\
   \footnotesize(a) Number of voters $N=50$ & \hspace{-2mm}\footnotesize(b) Number of voters $N=100$
   \end{tabular} 
 	\vspace{-4mm}
    \caption{\small{Accuracy of Algorithm~\ref{algo-record-level-func} (RLDP under functional mechanism) \textit{vs.} privacy parameter $\epsilon$ on synthetic dataset.}}
    \vspace{-3mm} 
	\label{syn-Fun}
\end{minipage}
~~~
\hspace{5mm}
\begin{minipage}{.48\textwidth}
\centering
  \begin{tabular}{cc} 
     \includegraphics[height=2.8cm]{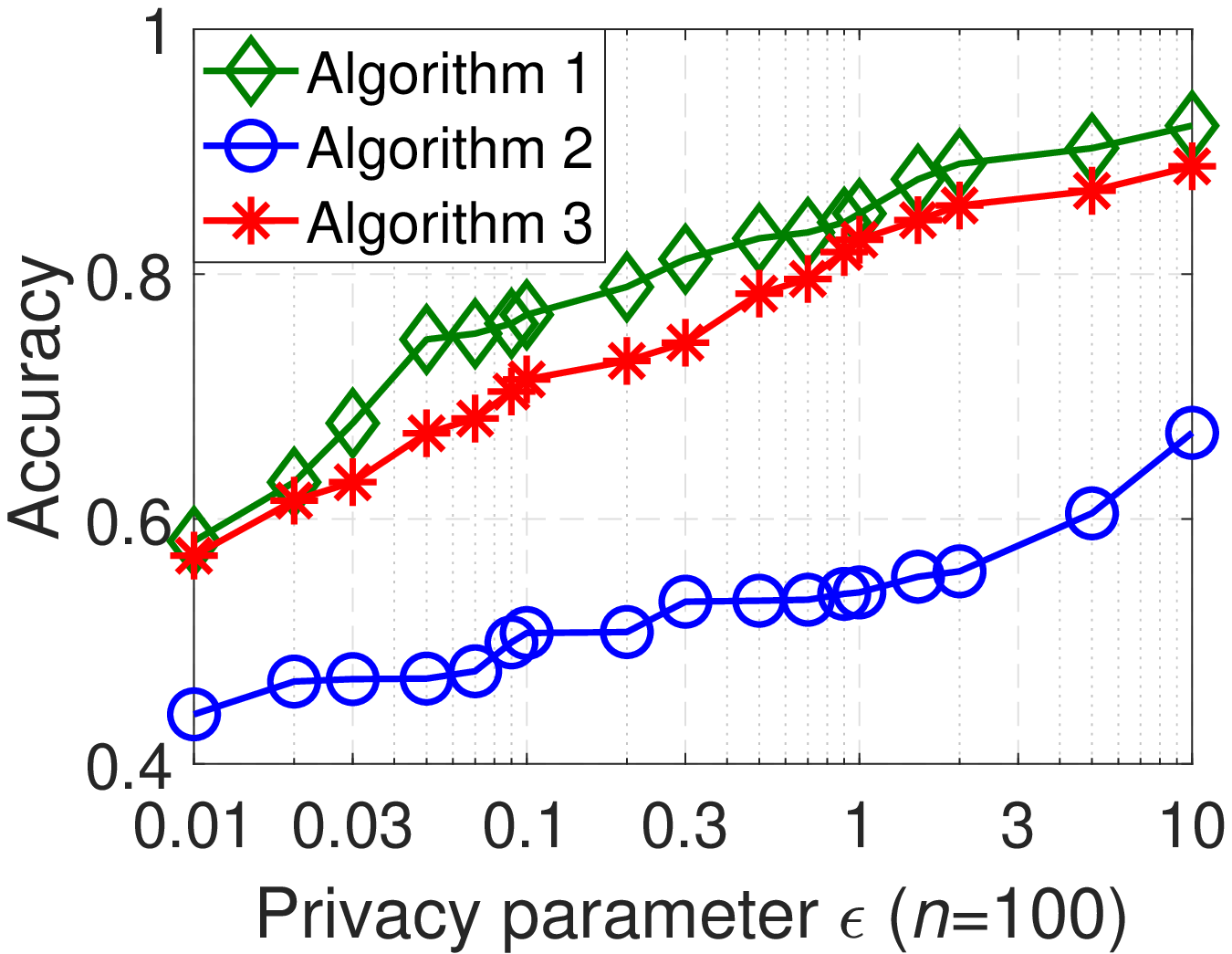} & \hspace{-3mm}\includegraphics[height=2.8cm]{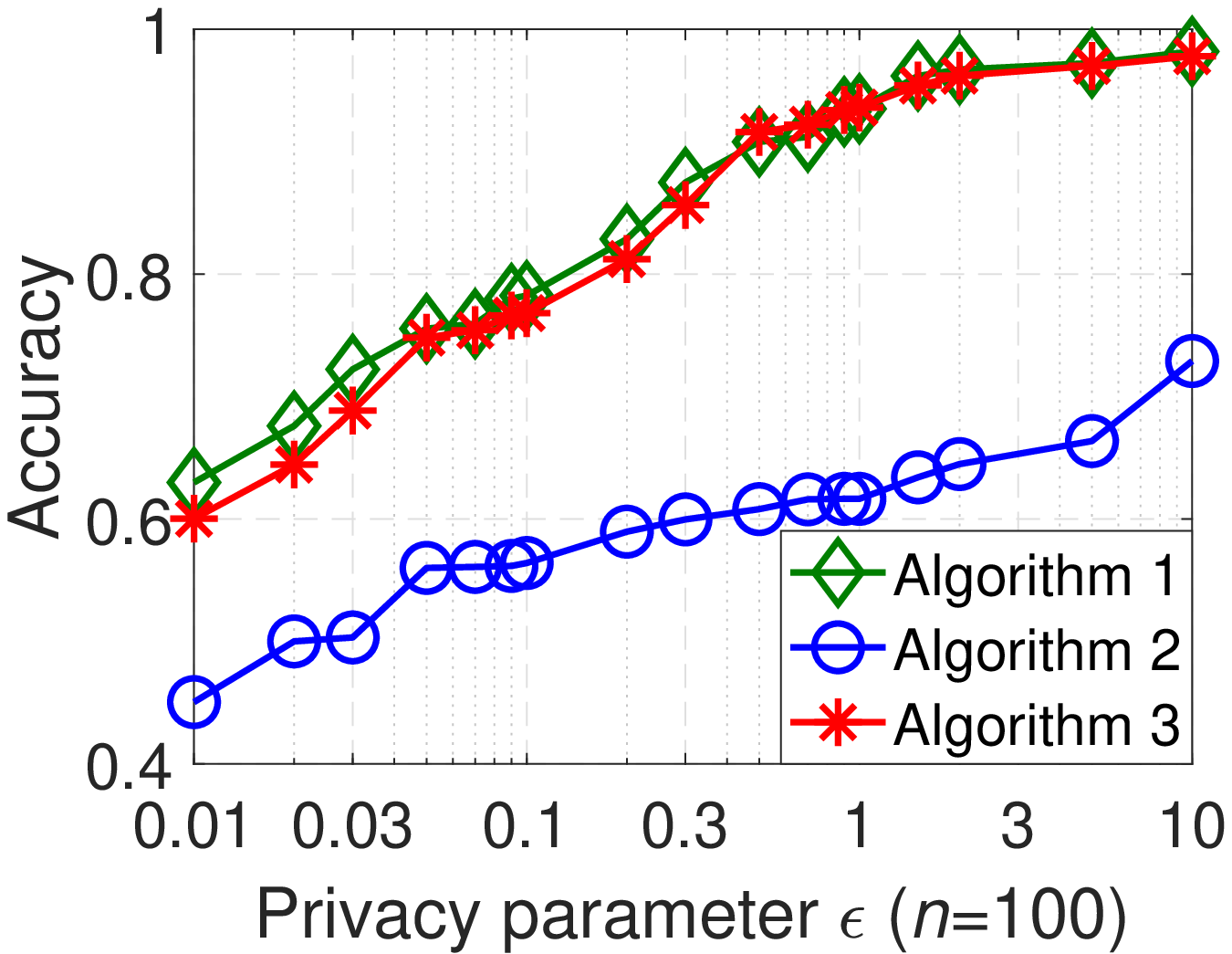} \vspace{-1mm} \\
   \footnotesize(a) Number of voters $N=50$ & \hspace{-2mm}\footnotesize(b) Number of voters $N=100$
   \end{tabular} 
 	\vspace{-4mm}
    \caption{\small{Comparisons of Algorithms~\ref{algo-user-level},~\ref{algo-record-level},~and~\ref{algo-record-level-func} on synthetic dataset.}}
    \vspace{-3mm} 
	\label{syn-compare123}
\end{minipage}
\end{figure*}

\section{Experimental Evaluation}\label{sec:experiments}

This section implements our proposed algorithms and evaluates their performance on synthetic datasets and a real-world dataset extracted from the  Moral Machine, respectively.

\begin{figure}
    \centering
    \begin{tabular}{cc} 
     \includegraphics[height=2.8cm]{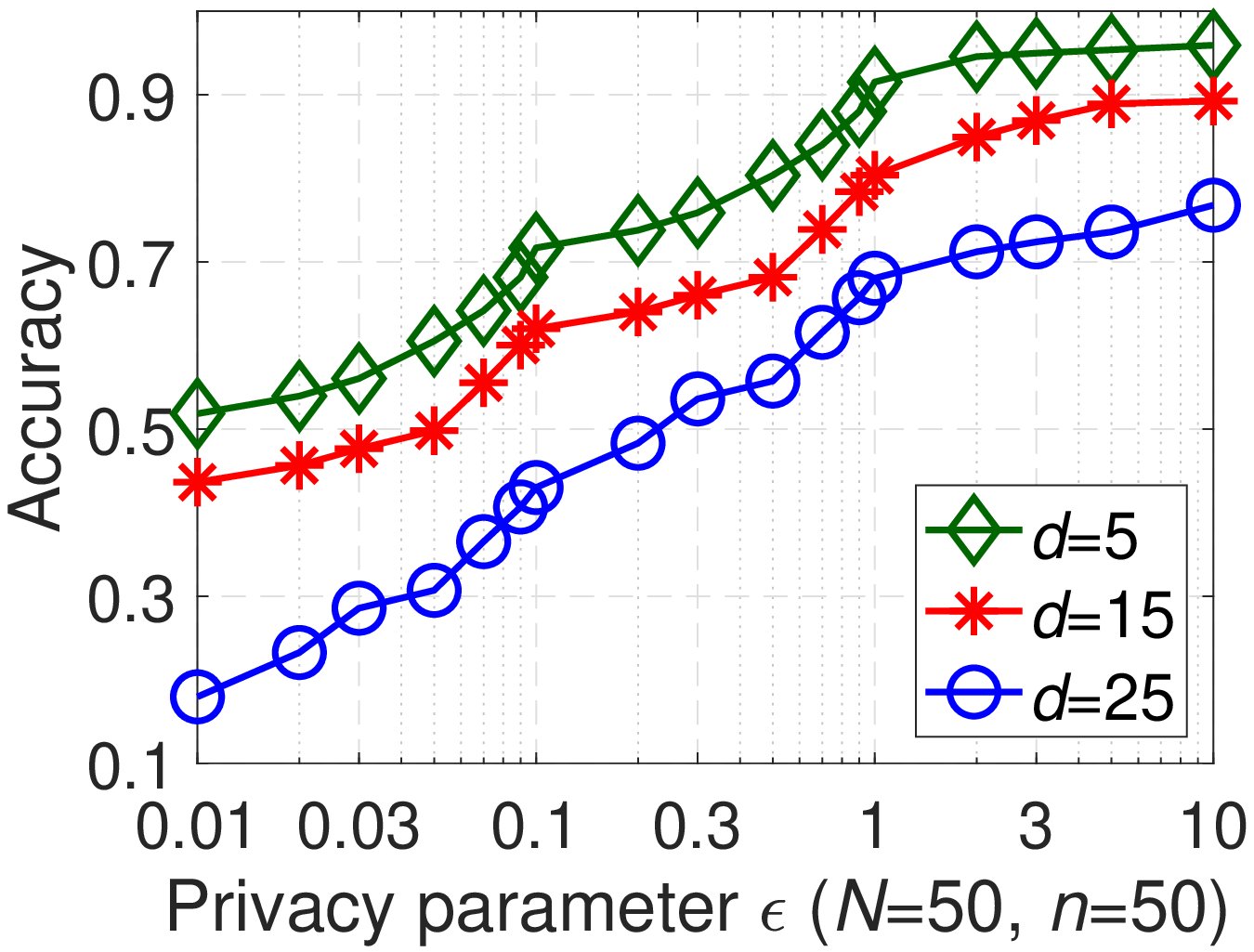} &\hspace{-5mm}
    \includegraphics[height=2.8cm]{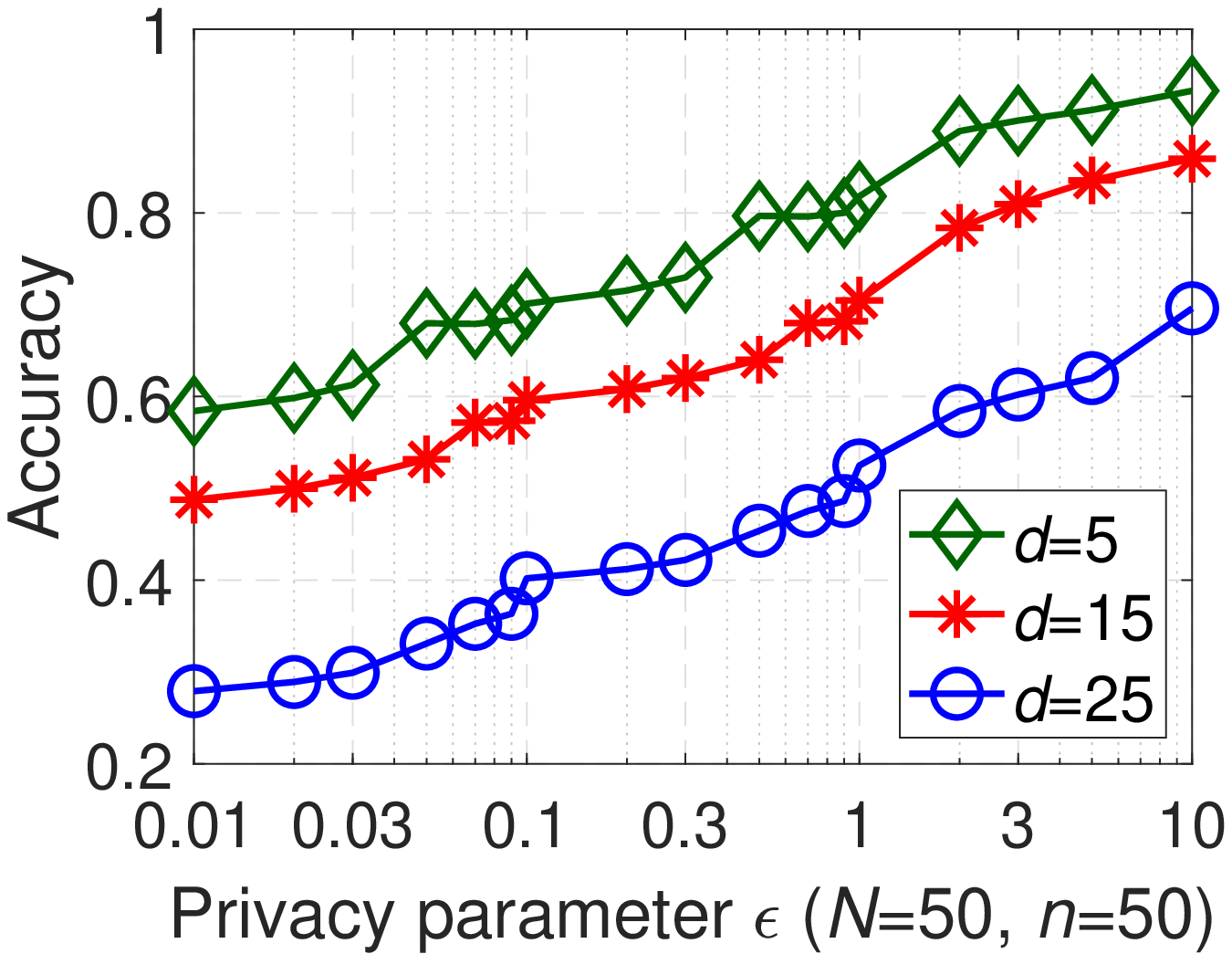} \vspace{-1mm} \\
    \footnotesize(a) Algorithm~\ref{algo-user-level} &  
    \footnotesize(b) Algorithm~\ref{algo-record-level-func}
   \end{tabular} 
 	\vspace{-4mm}
    \caption{\small{Accuracies of Algorithm~\ref{algo-user-level} and Algorithm~\ref{algo-record-level-func} \textit{vs.} dimension $d$ on synthetic dataset.}}
	\label{syn-vary-d}
	\vspace{-3mm} 
\end{figure}

\subsection{Synthetic Data}\label{sec:experiments-syn}

We first present the experimental results on synthetic data. Denote the true parameter of each voter $i$ as $\bm{\beta}_i$ which is sampled from Gaussian distribution $\mathcal{N}(\bm{m},\bm{I}_d)$, where each mean $m_j$ for $j\in \{1,2,\ldots, d\}$ is independently sampled from the uniform distribution $\mathcal{U}(-1,1)$,  and $\bm{I}_d$ is the $d \times d$ identity matrix with $d$ being the dimension size (i.e., the number of features). By default, we set $d=10$ when not specified. For each voter $i \in \{1,2,\ldots,N\}$, we generate its $n$ pairwise comparison via the following two steps (by default $n=50$ when not specified). First, we sample two alternatives $\bm{x}_1$ and $\bm{x}_2$ independently from the Gaussian distribution $\mathcal{N}(0,\bm{I}_d)$. Second, we sample their utilities $U_{\bm{x}_1}$ and $U_{\bm{x}_2}$ based on the Gaussian distribution $\mathcal{N}(\bm{\beta}_i^\top \bm{x}_1,\frac{1}{2})$ and $\mathcal{N}(\bm{\beta}_i^\top \bm{x}_2,\frac{1}{2})$, respectively. To evaluate the performance of our algorithms, we compute the accuracy which is defined as the fraction of test instances on which the true and noisy parameters give the same exact outcome. Besides, the parameter $B$ is set as $B=2$ in the following experiments. 

\begin{figure*}[ht]
\vspace{-3mm}
\hspace{-5mm}
\begin{minipage}{.72\textwidth}
    \centering
    \begin{tabular}{ccc} 
    \includegraphics[height=2.8cm]{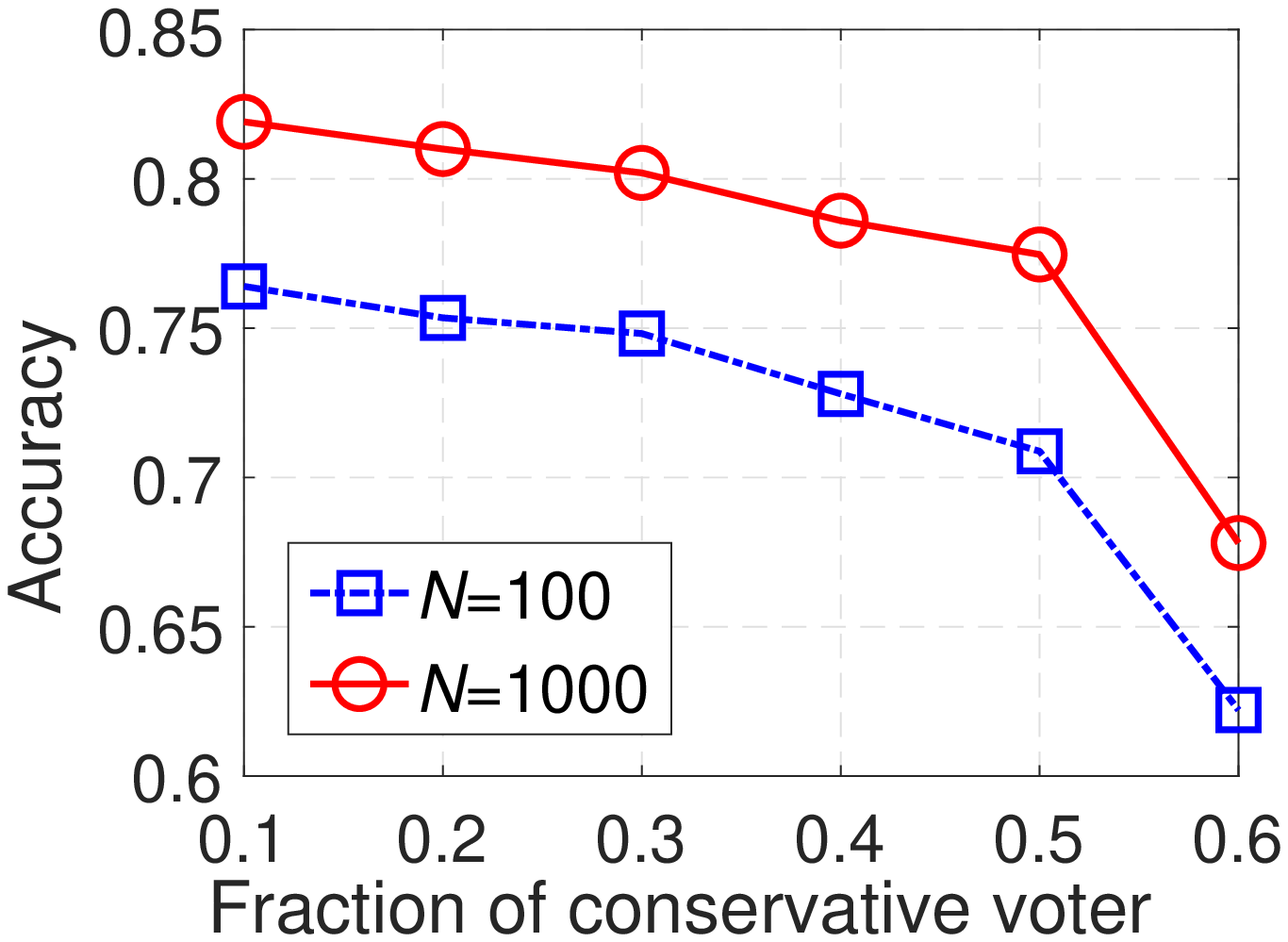}&
    \includegraphics[height=2.8cm]{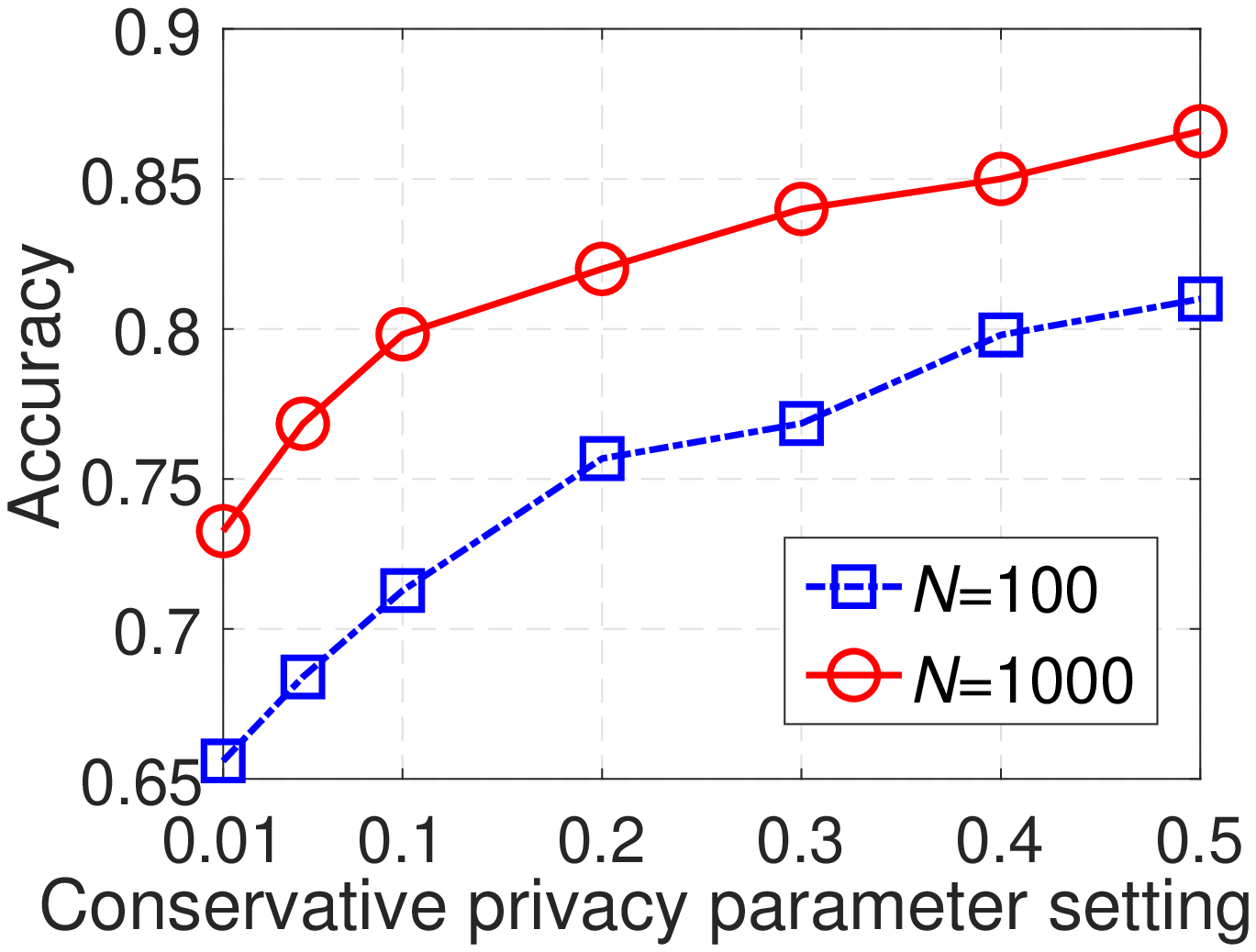} &
    \includegraphics[height=2.8cm]{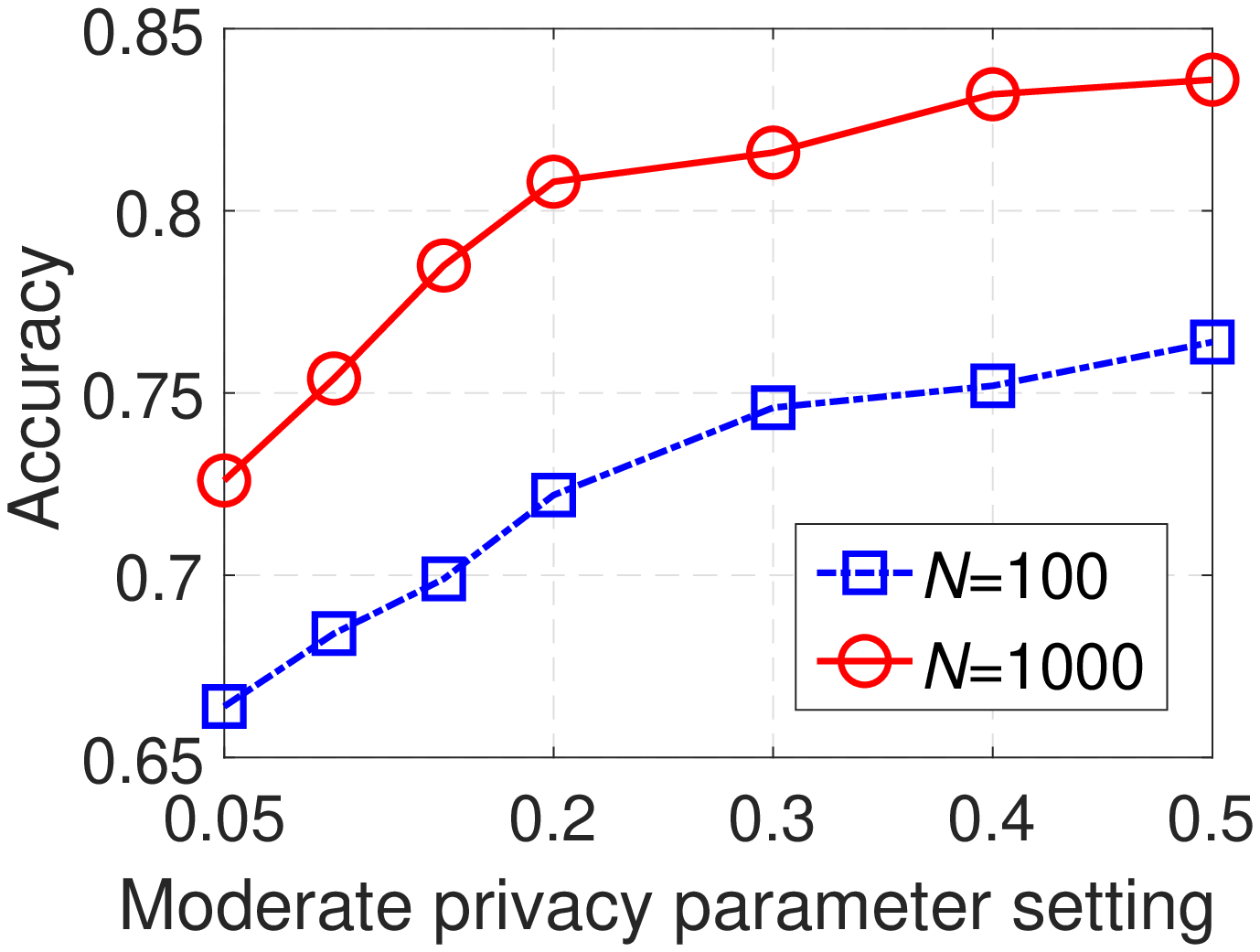} \vspace{-1mm} \\
    \footnotesize(a) Impact of $f_C$ & 
    \footnotesize(b) Impact of $\epsilon_C$ ($\epsilon_M=0.5$)&
    \footnotesize(c) Impact of $\epsilon_M$ ($\epsilon_C=0.01$) 
   \end{tabular} 
 	\vspace{-4mm}
    \caption{\small{Accuracies of Algorithm~\ref{algo-record-level-func} \textit{vs.} personalized privacy parameter settings with $\epsilon_L=1$.}}
    \vspace{-3mm}
	\label{syn-person-ec} 
\end{minipage}
~~~
\begin{minipage}{.24\textwidth}
    \centering
    \begin{tabular}{c} 
    \includegraphics[height=2.8cm]{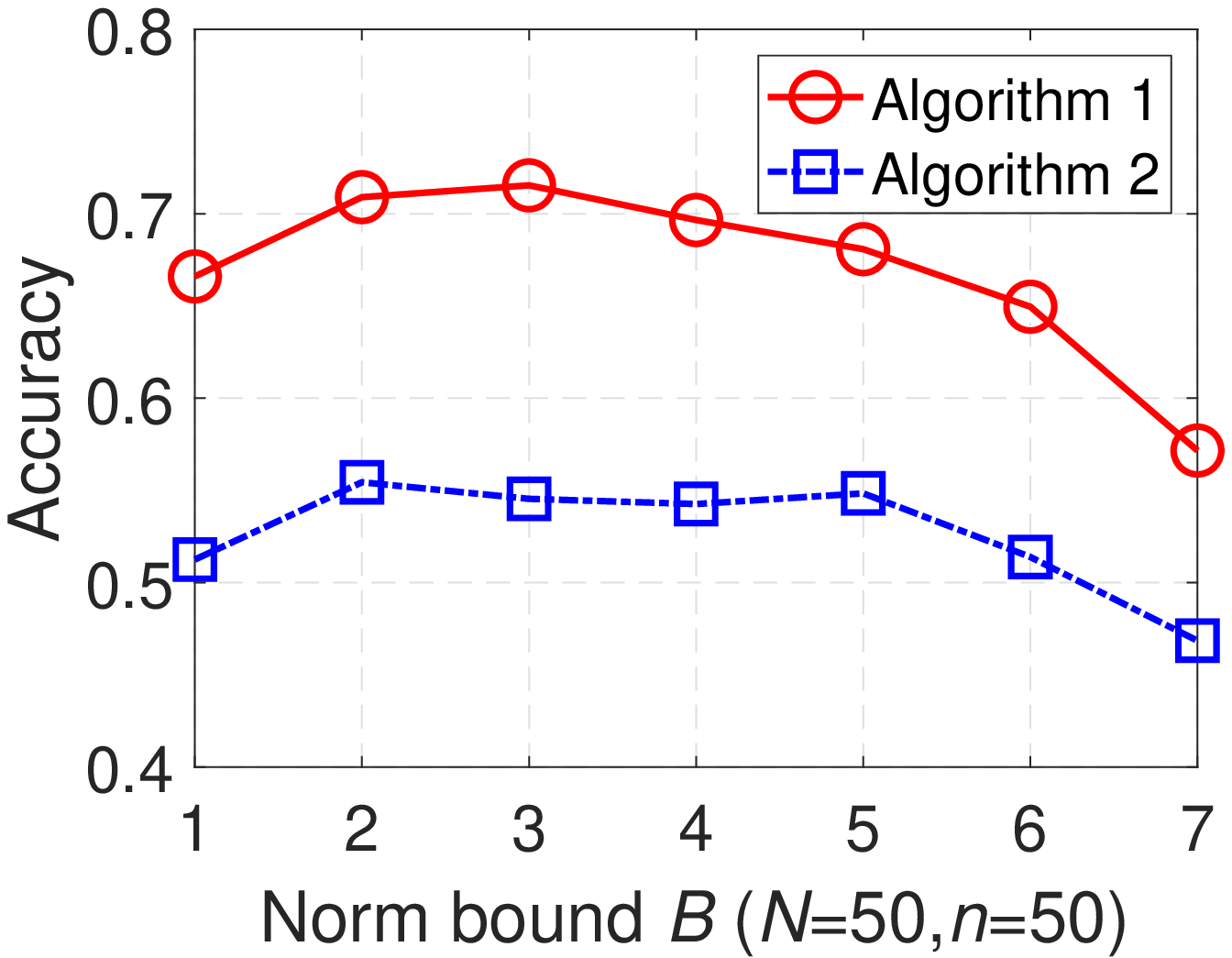}
   \end{tabular} 
 	\vspace{-5mm}
    \caption{\small{Accuracy
    \textit{vs.} norm bound values with $\epsilon=0.1$.}}
    \vspace{-3mm}
	\label{syn-norm-bound} 
\end{minipage}
\end{figure*}

The results in 
Figures~\ref{syn-UL},~\ref{syn-EL} and~\ref{syn-Fun} show high accuracies of our Algorithms~\ref{algo-user-level}, Algorithms~\ref{algo-record-level} and~\ref{algo-record-level-func} on the synthetic dataset with the privacy parameter $\epsilon$ varying from 0.01 to 10; specifically, we consider $\epsilon \in$ $ \{0.01,0.02,0.03,0.05,0.07,0.09,0.1,0.2,0.3,0.5,0.7,0.9,$ $1,2,3,5,10\}$. Note that we here use the same privacy parameter $\epsilon$ for all voters and will analyse the personalized privacy parameter specifically in the following. As shown in all figures, a smaller $\epsilon$ (i.e., higher privacy protection level) will lead to a lower accuracy.

Moreover, Figures~\ref{syn-UL}, \ref{syn-EL} and~\ref{syn-Fun} (all with $d=10$) show that for three Algorithms~\ref{algo-user-level}, \ref{algo-record-level} and \ref{algo-record-level-func}, a larger number of voters (i.e., $N$) or a larger number of records (i.e., $n$) leads to better accuracy. This is because our output is the average of each voter. Then with more voters, each one is relatively less sensitive. Therefore, it can achieve better accuracy (less noise) while keeping the same privacy.



We also compare the accuracies of three Algorithms~\ref{algo-user-level}, \ref{algo-record-level} and \ref{algo-record-level-func}.
As shown in Figure~\ref{syn-compare123}, Algorithms~\ref{algo-user-level} outperforms the other two algorithms. Also, both Algorithms~\ref{algo-user-level} and \ref{algo-record-level-func} have much better accuracy than Algorithm~\ref{algo-record-level}. This is because Algorithm~\ref{algo-record-level} realizes distributed privacy protection by adding Laplace noise, which ensures strong privacy guarantees, thus leading a relatively low data utility. In contrast, Algorithm~\ref{algo-record-level-func} also achieves distributed privacy protection by perturbing the objective function instead of parameters directly, which improves the data utility greatly. Therefore, Figure~\ref{syn-compare123} illustrates the superiority of functional mechanism over Laplace for record-level privacy protection with distributed perturbation. 


\textbf{Impact of the dimension $d$.}
Figure~\ref{syn-vary-d} shows the comparisons of the impact of parameter $d$ on the accuracy of the Algorithms~\ref{algo-user-level} and~\ref{algo-record-level-func}. Due to the space limitation, we no longer present the result of Algorithm~\ref{algo-record-level} here since it has limited data utility as shown in Figure~\ref{syn-compare123}. It can be seen from Figure~\ref{syn-vary-d} that the accuracy of both Algorithms~\ref{algo-user-level} and~\ref{algo-record-level-func} will be decreased with the increase of dimension $d$. For Algorithm~\ref{algo-user-level}, the larger $d$ means more information to be protected so the accuracy will be reduced. For Algorithm~\ref{algo-record-level-func}, the sensitivity will increase with the increase of $d$, thus leading to the reduction of the accuracy when $d$ becomes larger.

\textbf{Impact of the personalized privacy parameters.} We also implement personalized privacy parameter choices for RLDP (i.e., Algorithm~\ref{algo-record-level-func}) and conduct extensive experiments to show the impact of the personalized privacy parameter on accuracies. Although Algorithm~\ref{algo-record-level} can also achieve personalized privacy protection, we no longer present the result of Algorithm~\ref{algo-record-level} here since it has limited data utility and due to the space limitation.

We set personalized privacy parameter specifications based on the findings from studies in \cite{jorgensen2015conservative,acquisti2005privacy}. 
The voters are randomly divided into three groups: (\textit{i}) \textit{conservative}, (\textit{ii}) \textit{moderate}, and (\textit{iii}) \textit{liberal}, representing voters with high/medium/low privacy concern, respectively.
The fraction of voters in three groups are
$f_C$, $f_M$, and $f_L$, where $f_C+f_M+f_L=1$. We set default values as $f_C = 0.54$, $f_M=0.36$, and $f_L=0.1$, which are chosen based on the findings reported in \cite{acquisti2005privacy}.
The privacy parameters $\epsilon$ of voters who belong to conservative or moderate groups are chosen uniformly at random from $[\epsilon _C, \epsilon _M]$ or $[\epsilon _M, \epsilon _L]$ correspondingly (and rounding to the nearest hundredth), where $\epsilon _C\in\{ \underline{0.01},0.05,0.1,0.2,\cdots,0.5\}$, $\epsilon _M\in\{0.05,0.1,0.15,\underline{0.2},\cdots,0.5\}$, and $\epsilon _L= \underline{1}$, and the default values are underlined.

\begin{figure*}[h]
\vspace{-1mm}
\centering
    \begin{tabular}{cccc} 
    \hspace{-4mm}\includegraphics[height=2.8cm]{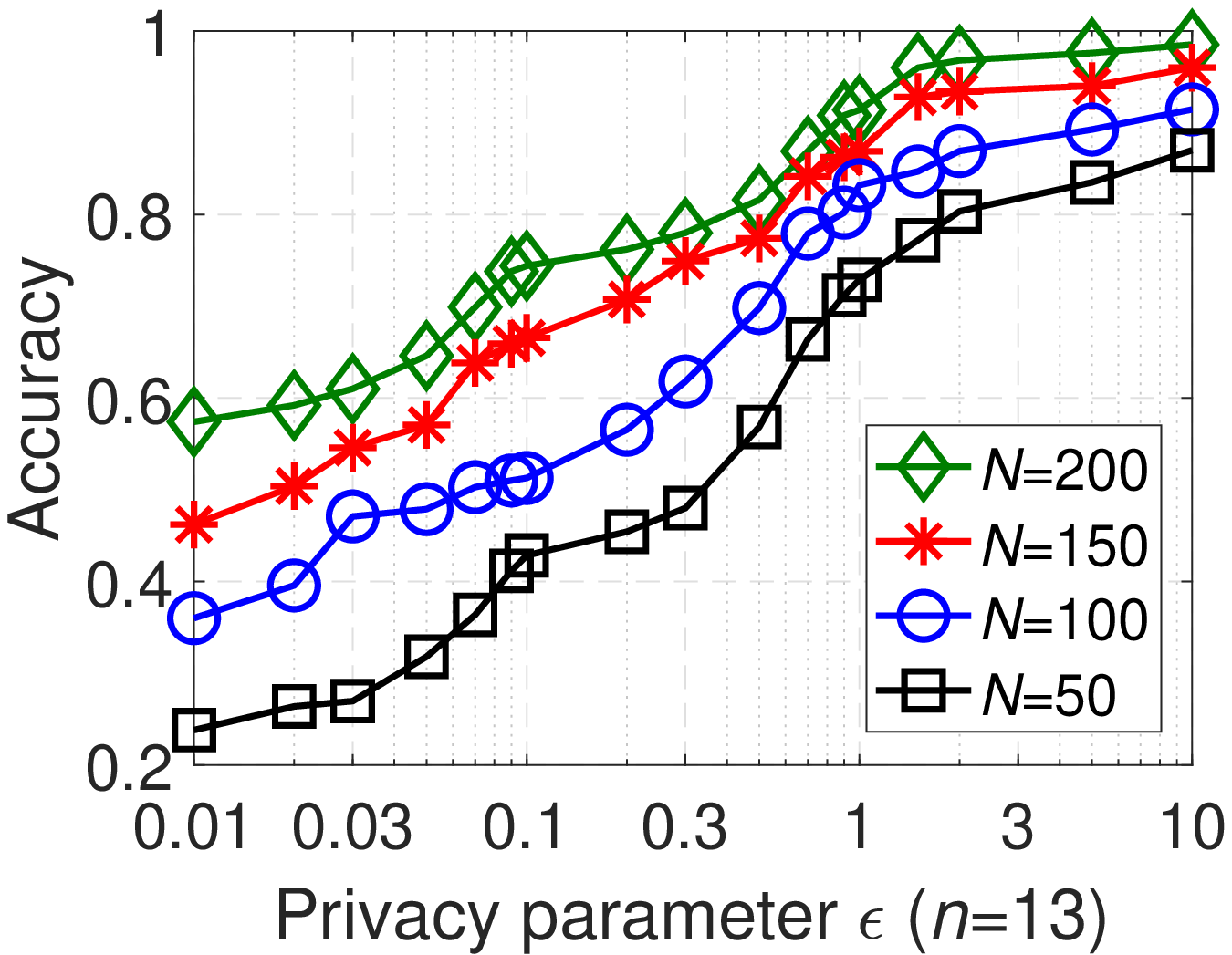} &
    \hspace{-2mm}\includegraphics[height=2.8cm]{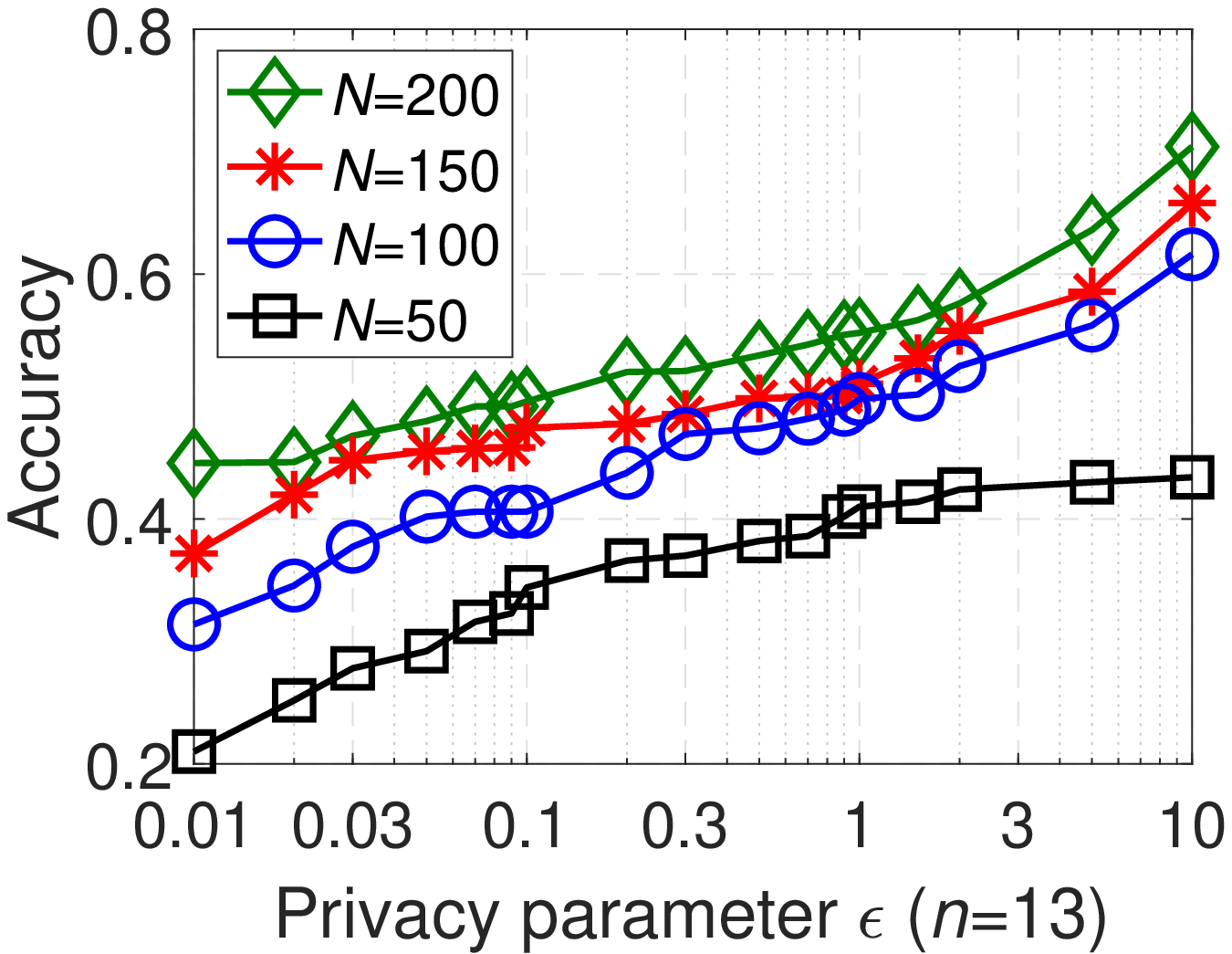} &
    \hspace{1mm}\includegraphics[height=2.8cm]{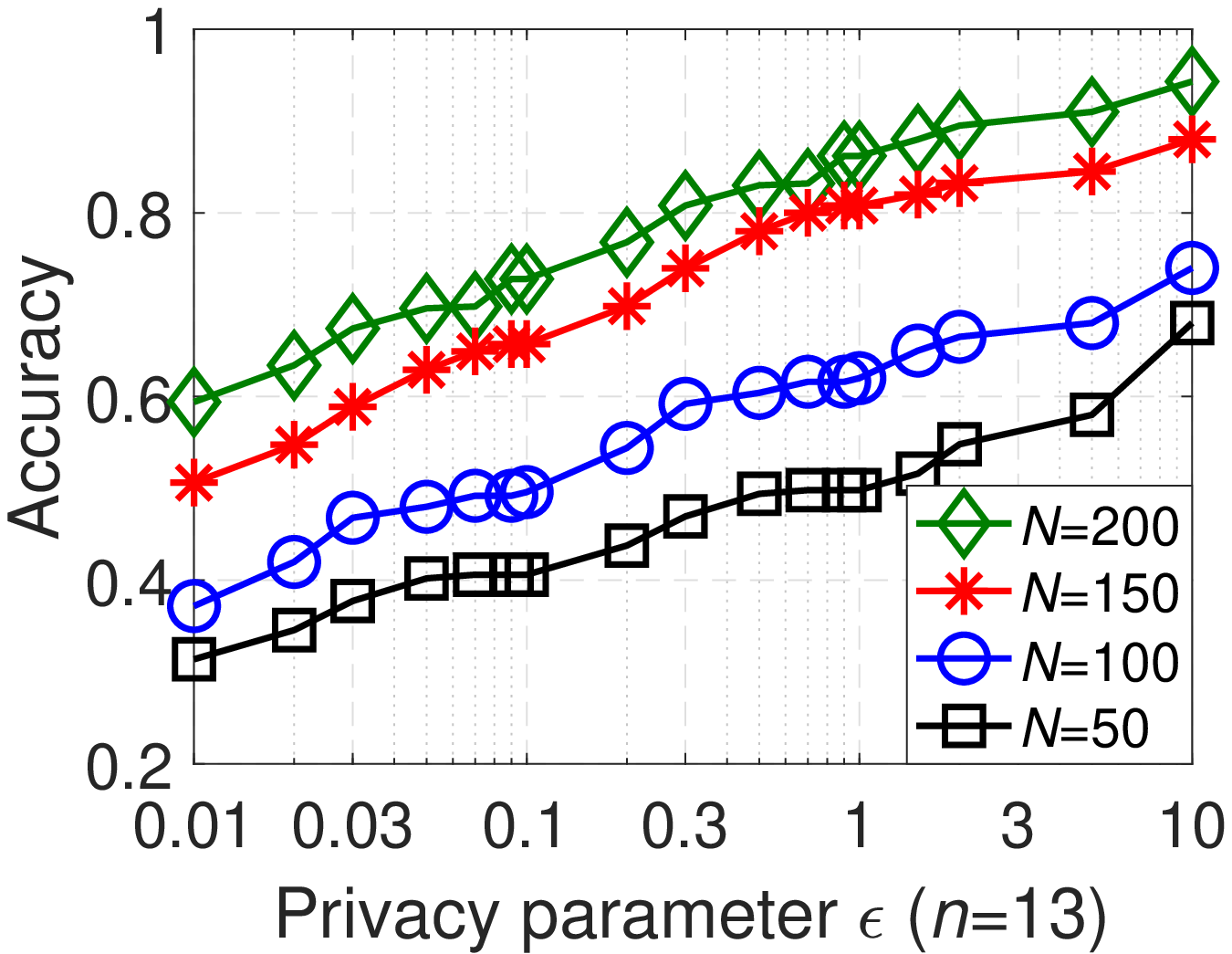} &
    \hspace{1mm}\includegraphics[height=2.8cm]{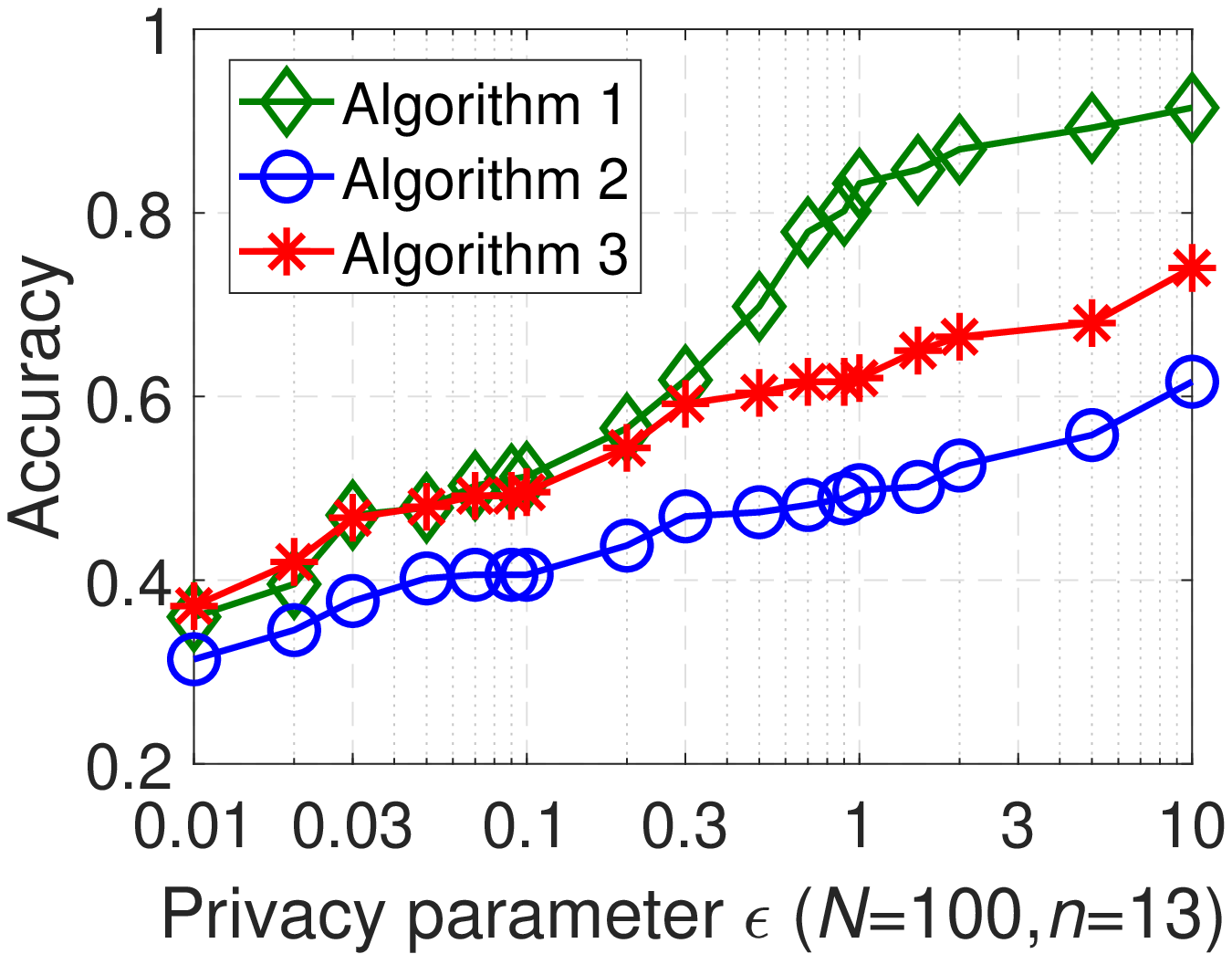} \vspace{-1mm} \\
    \footnotesize(a) Algorithm~\ref{algo-user-level}&
    \footnotesize(b) Algorithm~\ref{algo-record-level} &
    \footnotesize(c) Algorithm~\ref{algo-record-level-func} &
    \hspace{-5mm}\footnotesize(d) Comparisons of Algorithms~\ref{algo-user-level},~\ref{algo-record-level}~and~\ref{algo-record-level-func}
   \end{tabular} 
 	\vspace{-4mm}
    \caption{Accuracies of Algorithms~\ref{algo-user-level},~\ref{algo-record-level}~and~\ref{algo-record-level-func} versus privacy parameter $\epsilon$ on real dataset from the Moral Machine.}
    \vspace{-3mm}
	\label{real}
\end{figure*}

Then, we implement Algorithm~\ref{algo-record-level-func} by setting personalized privacy parameters as specified above.
Figure~\ref{syn-person-ec}(a) shows the accuracy of Algorithm~\ref{algo-record-level-func} by varying the fraction of conservative voters $f_C$ from 0.1 to 0.6 by setting $n=50$ and $d=10$. Note that for each set of $f_C$, the fraction of moderate voter is fixed as $f_M = 0.36$ and the fraction of liberal voters is equal to $1-(f_C+f_M)$. We can observe from Figure~\ref{syn-person-ec}(a) the accuracy will be reduced when increasing $f_C$. This is because increasing $f_C$ will increase the number of conservative voters (prefer for choosing small $\epsilon$) but decrease the number of liberal voters (prefer for choosing large $\epsilon$), thus leading to more noise injections and accuracy reductions.

Furthermore, Figures~\ref{syn-person-ec}(b) and \ref{syn-person-ec}(c) present experimental results on the impact on accuracy when varying the range of $\epsilon_C$ or $\epsilon_M$, respectively. Figure~\ref{syn-person-ec}(b) shows the accuracy of Algorithm~\ref{algo-record-level-func} when varying $\epsilon_C$ from 0.01 to 0.5 by fixing moderate privacy parameter as $\epsilon_M=0.5$ (rather than the default value of $\epsilon_M=0.2$ to ensure $\epsilon_C \leq \epsilon_M$ in all cases). It can be seen that the larger $\epsilon_C$ or $N$ will induce better accuracy. Besides, Figure~\ref{syn-person-ec}(c) shows the accuracy of Algorithm~\ref{algo-record-level-func} when changing $\epsilon_M$ from 0.05 to 0.5 by fixing
$\epsilon_C=0.01$. We can observe again that larger $\epsilon_M$ or $N$ will lead to better accuracy. What's more, the accuracy of Figure~\ref{syn-person-ec}(b) is higher than that of Figure~\ref{syn-person-ec}(c). This is because the conservative privacy parameter in Figure~\ref{syn-person-ec}(c) always equals to 0.01 and the fraction of conservative voters is bigger than that of moderate voters.

\textbf{Impact of the norm bound $B$.} In addition, we also analyze the impact of norm bound $B$ on the accuracies of Algorithms~\ref{algo-user-level} and \ref{algo-record-level}, as shown in Figure~\ref{syn-norm-bound}. It can be observed that the accuracy will be relatively higher when norm bound around at $2$ or $3$.

\subsection{Moral Machine Data}\label{sec:experiments-real}

The Moral Machine data contain the moral decisions by voters about ethical dilemmas faced by autonomous vehicles. Each voter is asked to view 13 scenarios (i.e. $n=13$), each of which includes two possible moral dilemma alternatives. Each alternative is characterized by 23 features (i.e. $d=23$). We sample $200$ voters' data ($N=200$) from the Moral Machine dataset to evaluate our algorithms.

Figures~\ref{real}(a), \ref{real}(b) and \ref{real}(c) plot the accuracies of our Algorithms~\ref{algo-user-level},\ref{algo-record-level} and~\ref{algo-record-level-func} on the real dataset for privacy parameters $\epsilon$ from 0.01 to 10 and for different numbers of voters (i.e., $N$), respectively. The privacy parameter specifications follow the same settings of the synthetic dataset in Section~\ref{sec:experiments-syn}. As expected, larger $N$ or  larger $\epsilon$ will induce  better accuracy for all three Algorithms~\ref{algo-user-level},\ref{algo-record-level} and~\ref{algo-record-level-func}.

In addition, Figure~\ref{real}(d) shows the comparisons of three Algorithms~\ref{algo-user-level},\ref{algo-record-level} and~\ref{algo-record-level-func} with $N=100,n=13$. It can observe from Figure~\ref{real}(d) that Algorithm~\ref{algo-user-level} relatively outperforms Algorithms~\ref{algo-record-level}~and~\ref{algo-record-level-func} and the Algorithm~\ref{algo-record-level} has the lowest accuracy in all cases. This is because Algorithm~\ref{algo-record-level} follows a strong privacy definition with distributed Laplace perturbation, which leads to larger noise addition. In contrast, Algorithm~\ref{algo-record-level-func} also achieves distributed privacy protect while ensuring  a good data utility since it perturbs objective functions instead of parameters. Thus, this demonstrates the superiority of functional mechanism over Laplace mechanism for record-level privacy protection with distributed perturbation. 



Furthermore, extensive experiments for personalized parameter settings are also conducted on the Moral Machine data, which follows the same privacy parameter specifications as Section~\ref{sec:experiments-syn}. Figure~\ref{real-person} shows the accuracies of the Algorithm~\ref{algo-record-level-func} under different personalized privacy parameters. Two figures in  Figure~\ref{real-person} present the impact of the conservative privacy parameters $\epsilon_C$ (resp., moderate privacy parameters $\epsilon_M$) on accuracies by varying $\epsilon_C$ from 0.01 to 0.5 with $\epsilon_M=0.5$ (reps., $\epsilon_M$ from 0.05 to 0.5 with $\epsilon_C=0.01$). Both figures show that larger $\epsilon_C$ or $\epsilon_M$ will lead to higher accuracy.

\begin{figure*}[h]
\centering
\hspace{-2mm}
\begin{minipage}{.48\textwidth}    
    \begin{tabular}{cc} 
     \includegraphics[height=2.8cm]{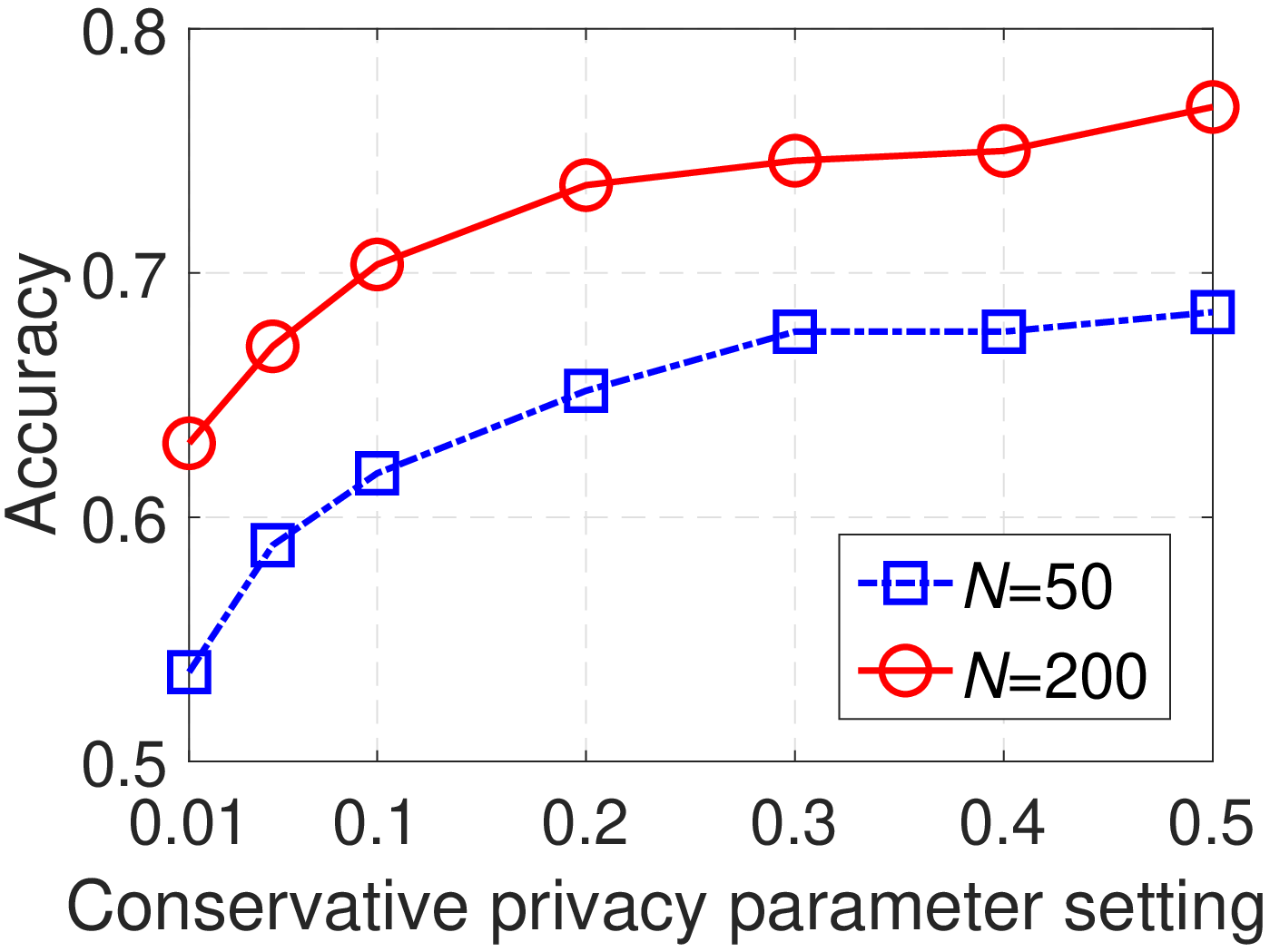} &\hspace{-2mm}
    \includegraphics[height=2.8cm]{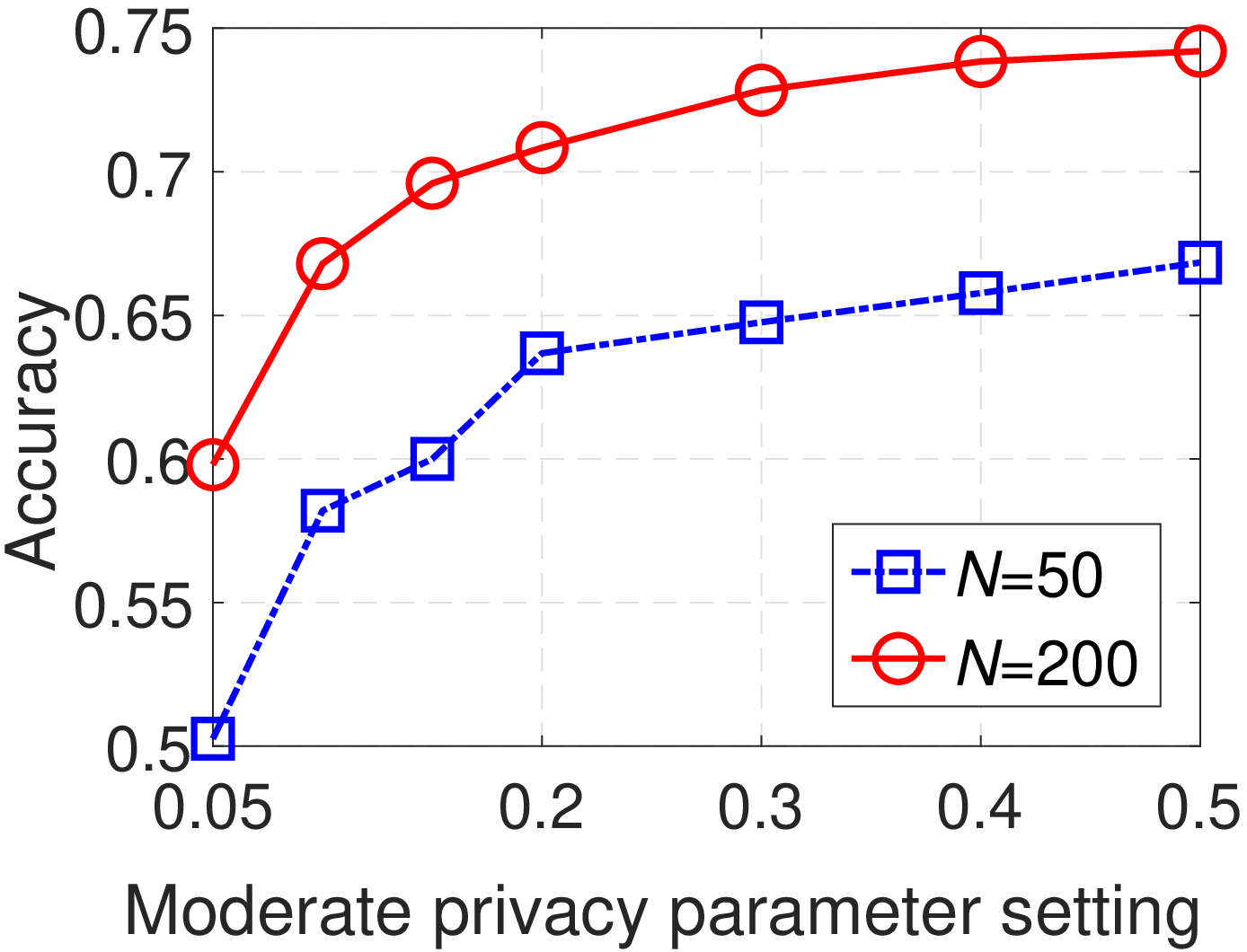}  \vspace{-1mm}\\
    \footnotesize(a) Impact of $\epsilon_C$ ($\epsilon_M=0.5$) & 
    \footnotesize(b) Impact of $\epsilon_M$ ($\epsilon_C=0.01$)
   \end{tabular} 
 	\vspace{-4mm}
        \caption{\small{Accuracies of Algorithm~\ref{algo-record-level-func} \textit{vs.} personalized privacy parameters with $\epsilon_L=1$ on real dataset from the Moral Machine.}}
    \vspace{-3mm}
	\label{real-person}
\end{minipage}
~~~
\hspace{3mm}
\centering
\begin{minipage}{.48\textwidth}
    \begin{tabular}{cc} 
     \includegraphics[height=2.8cm]{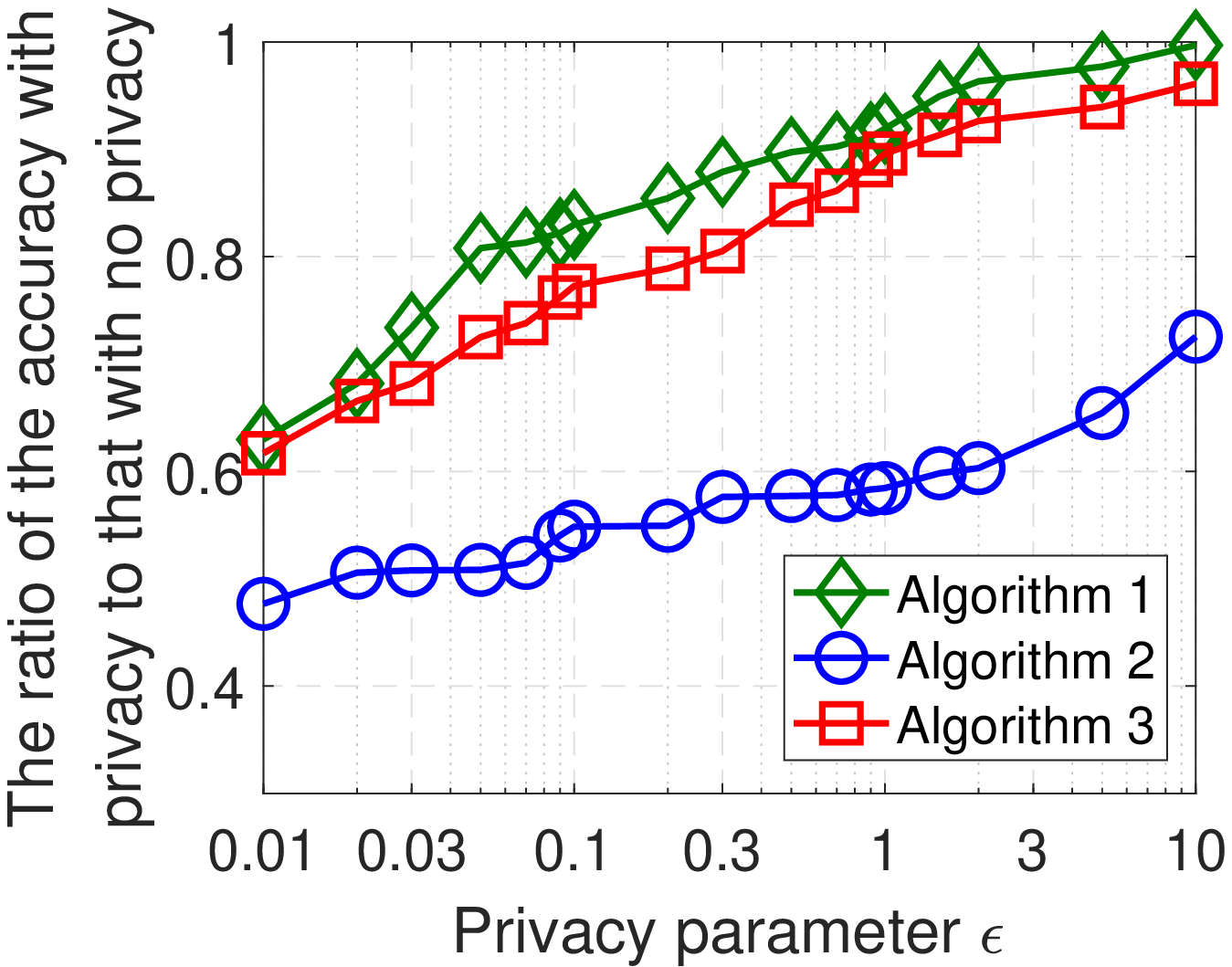} &\hspace{-3mm}
    \includegraphics[height=2.8cm]{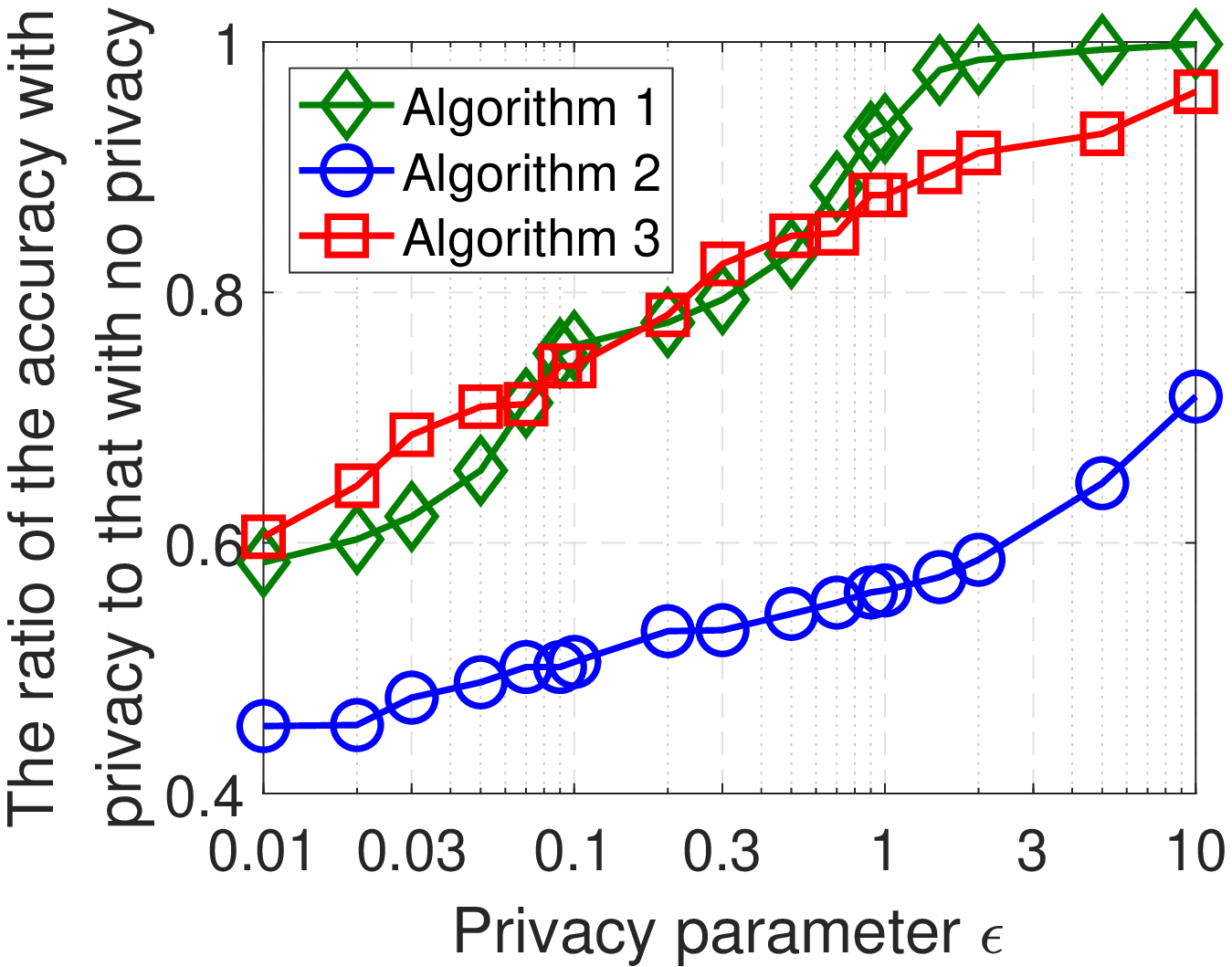} \vspace{-1mm} \\
    \footnotesize(a) Synthetic dataset &   \footnotesize(b) Real dataset 
   \end{tabular} 
 	\vspace{-4mm}
        \caption{\small{Comparisons on the accuracies between our privacy-preserving algorithms with that with no privacy.}}
    \vspace{-3mm}
	\label{compare-no-privacy}
\end{minipage}
\end{figure*}

\subsection{Comparisons with No Privacy}
To further show the effectiveness of our privacy-preserving algorithms, we present the ratio of the accuracy of our algorithms to the baseline method that with no privacy protection (i.e., solution of Noothigattu \textit{et al.}~\cite{noothigattu2018voting}). As shown in Figure~\ref{compare-no-privacy}, Algorithms~\ref{algo-user-level}~and~\ref{algo-record-level-func} hold better accuracy than Algorithm~\ref{algo-record-level} on both synthetic dataset (with $N=50$ and $n=100$, Noothigattu \textit{et al.}'s method has $92.4\%$ accuracy) and real dataset (with $N=200$ and $n=13$, Noothigattu \textit{et al.}'s method has $98.2\%$ accuracy). In particular, we can observe that Algorithms~\ref{algo-user-level}~and~\ref{algo-record-level-func} can ensure greater than $80\%$ accuracy when $\epsilon>0.3$ and greater than $90\%$ accuracy when $\epsilon>1$.

\section{Discussions and Future Directions}\label{sec:discussion}

\textbf{Choice of $\epsilon$.} 
Our algorithms adopt the formal notion of $\epsilon$-DP in which smaller $\epsilon$ means stronger privacy protection but less utility. In the literature, typical $\epsilon$ values chosen for experimental studies are mostly between $0.1$ and $10$~\cite{jorgensen2015conservative,abadi2016deep,wang2019collecting}. In the industry, both Google Chrome and Apple iOS have implemented $\epsilon$-DP (in particular, a specific variant called $\epsilon$-LDP~\cite{duchi2013local}). In Google Chrome~\cite{erlingsson2014rappor}, $\epsilon$ per each datum is at most $2$, and $\epsilon$ per each user is at most $9$. Apple claims that $\epsilon$ per each datum per day in iOS is between $1$ and $2$, but a recent study~\cite{tang2017privacy} shows that it can be as high as $16$.




\textbf{Privacy Protection for Crowdsourced Data Collection for Fairness Studies of Machine Learning.}
Fairness of machine learning algorithms has received much interest recently~\cite{srivastava2019mathematical}. Among various fairness notions, demographic parity means equalizing the percentage of people who are predicted by the algorithm to be positive across different social groups (e.g., racial groups)~\cite{dwork2012fairness}. Srivastava\textit{~et al.~}\shortcite{srivastava2019mathematical} have collected a large crowdsourced dataset to capture people's perception of fairness of machine learning. As a user's perception of fairness may contain sensitive information, a future direction in the same vein as our current paper is to preserve the privacy of users in the aggregation of fairness preferences.



\section{Conclusion}\label{sec:conclusion}

This paper incorporates privacy protection into crowdsourced data collection used to guide ethical decision-making by AI. Based on the formal notion differential privacy, we proposed four different privacy protection paradigms with the consideration of voter-/record-level privacy protection and centralized/distributed perturbation.
Then, we specifically propose three algorithms which achieve the above four granularities of privacy protection paradigms. We have proved the privacy guarantees and the data utilities of our algorithms. To show the effectiveness of our algorithms, extensive experiments have been conducted on both synthetic datasets and a real-world dataset from the Moral Machine project.

\begin{acks}
The authors would like to thank the anonymous reviewers for their valuable comments. This work was partly supported by Natural Science Foundation of China (Nos. 61572398, 61772410, 61802298 and U1811461), the Fundamental Research Funds for the Central Universities (No. xjj2018237), and the China Postdoctoral Science Foundation (No. 2017M623177). Jun Zhao's research was supported by Nanyang Technological University (NTU) Startup Grant M4082311.020, Alibaba-NTU Singapore Joint Research Institute (JRI), Singapore Ministry of Education Academic Research Fund (AcRF) Tier 1 RG128/18, and NTU-WASP Joint Project M4082443.020. Han Yu's research was supported by the Nanyang Assistant Professorship (NAP) and the Joint NTU-WeBank Research Centre of Eco-Intelligent Applications (THEIA), Nanyang Technological University, Singapore.
\end{acks}




\bibliographystyle{ACM-Reference-Format}
\balance
\bibliography{bibfile}



\end{document}